 \documentclass[journal]{IEEEtran}

\usepackage[usenames,dvipsnames]{pstricks}

\usepackage[english]{babel}
\usepackage{verbatim}
\usepackage{psfrag}
\usepackage{url}

\usepackage{multirow}
\usepackage{rotating}
\usepackage{graphicx}
\usepackage{makeidx,shortvrb,latexsym}
\usepackage{fancyhdr}

\usepackage{indentfirst}

\usepackage{subfigure}
\usepackage{amsmath,amssymb,lineno,enumerate,amsthm,amsfonts}
\usepackage{tikz}
\usetikzlibrary{arrows,%
                shapes,positioning}

                \usepackage{tikz}

\usetikzlibrary{arrows,%
                petri,%
                topaths}%
\usepackage{tkz-berge}
\input xy
\xyoption{all}
\newcommand{\minrk}{\mathrm{minrk}}
\newcommand{\FF}{\mathbb F}
\newcommand{\Fq}{{\mathbb F}_q}

\def\C{\mathbf C}

\def\cB{\mathcal B}
\def\cC{\mathcal C}
\def\cD{\mathcal D}
\def\cE{\mathcal E}

\def\cG{\mathcal G}
\def\cI{\mathcal I}
\def\cH{\mathcal H}

\def\cL{\mathcal L}
\def\cM{\mathcal M}

\def\cO{\mathcal O}
\def\cP{\mathcal P}
\def\cR{\mathcal R}
\def\cS{\mathcal S}

\def\cV{\mathcal V}

\def\cX{\mathcal X}

\def\deg{\mbox{\rm deg}}

\def\min{{\rm min}}

\def\dim{\mbox{\rm dim}}
\def\supp{\mbox{\rm Supp}}

\newcommand{\ga}{\alpha}
\newcommand{\gb}{\beta}

\newcommand{\gl}{\lambda}
\newcommand{\gk}{\kappa}

\newcommand{\gt}{\tau}

\def\supp{{\rm Supp}}

\newcommand{\bc}{{\bf c}}

\newcommand{\bv}{{\bf v}}

\newcommand{\rk}{{\rm{rank}}}

\newcommand{\fq}{{\mathbb F}_q}

\newcommand{\ba}{{\bf a}}

\newcommand{\bb}{{\bf b}}

\newtheorem{theorem}{Theorem}[section]
\newtheorem{proposition}[theorem]{Proposition}
\newtheorem{lemma}[theorem]{Lemma}
\newtheorem{corollary}[theorem]{Corollary}

\newtheorem{remark}[theorem]{Remark}

{\theoremstyle{definition}
\newtheorem{definition}[theorem]{Definition}
\newtheorem{example}[theorem]{Example}

\newtheorem*{proposition*}{Proposition}
\newtheorem*{corollary*}{Corollary}
\newtheorem*{lemma*}{Lemma}

\ifCLASSINFOpdf
\else
\fi
\def\INSe#1{{\color{black}#1}}
\def\INSr#1{{\color{black}#1}}
\hyphenation{op-tical net-works semi-conduc-tor}

\begin{document}
%
\title{Bounding the optimal rate of the ICSI and ICCSI problem}
%
%
%

\author{Eimear~Byrne,
        and~Marco~Calderini. 
\thanks{School of Mathematical Sciences, 
University College Dublin, Ireland.}
\thanks{e-mail: ebyrne@ucd.ie}
\thanks{Department of Mathematics,
 University of Trento, Italy.}
\thanks{email: marco.calderini@unitn.it}
\thanks{Research supported by ESF COST Action IC1104}
}
\maketitle

\begin{abstract}
In this work we study both the index coding with side information (ICSI) problem \INSr{introduced by Birk and Kol in 1998} and the more general problem of index coding with coded side information (ICCSI), \INSr{described} by Shum {\em et al} in 2012. 
We \INSe{estimate} the optimal rate of an instance of the index coding problem. In the ICSI problem case, we characterize those digraphs having min-rank one less than their order and we give an upper bound on the min-rank of a hypergraph whose incidence matrix can be associated with that of a 2-design. \INSr{Security aspects} are discussed in the particular case when the design is a projective plane.
For the coded side information case, we extend the graph theoretic upper bounds given by Shanmugam {\em et al} in 2014 on the optimal rate of index code. 
\end{abstract}

%
%
%
%

\section{Introduction}
\label{intro}

Since its introduction in \cite{662940}, the problem of index coding has been generalized in a number of directions \cite{GCG-cd-art-Alon2008,GCG-cd-art-Bar-Yossef2006,6034005,GCG-cd-art-dau2012security,GCG-cd-art-dau2011index,GCG-cd-art-shum2014}. It is a problem that has aroused much interest in recent years; from the theoretical perspective, its equivalence to network coding  has established it as an important area of network information theory \cite{el2010index,effind15}. In the classical case, a central broadcaster has a data file $x \in \fq^n$. There are $n$ users each of whom already possesses some subset of components of $x$ as its side-information and each of whom requests some component $x_i$ of the file. The index coding problem is to determine the minimum number of transmissions required so that the demands of all users can be met, given that data may be encoded prior to broadcast.
This problem can be associated with a directed graph, or a hypergraph if the case is extended to consider a scenario of $m>n$ users. Several authors have given various bounds on the length of an index code, which refers to the number of transmission\INSr{s} used to meet clients' demands for a given instance of the problem. It is well known that for the case of linear index coding, the min-rank of the associated side-information graph is the minimal number of broadcasts required. In \cite{shanmugam2014graph}, the authors give several graph theoretic upper bounds based on linear programming.
In \cite{GCG-cd-art-shum2014} the authors describe the scenario of linear index coding with coded side information. In this model, users may request a linear combination of the data held by the sender and are assumed to each have some set of linear combinations of the data packets. One motivation for this more general model is that it may serve a larger number of applications than the case for uncoded side-information, such as broadcast relay networks and wireless distributed storage systems. The set-up in \cite{GCG-cd-art-shum2014} does not have an obvious representation in the form of a side-information hypergraph. However, as we show here, practically all the results of \cite{shanmugam2014graph} can be extended to this case.

In this paper we present new bounds on the optimal rate for different instances of the index coding problem. For the case of uncoded side information the problem will be referred to as an index coding with side information (ICSI) problem. For the case of encoded side information we will describe this as an ICCSI instance. In the first part we give bounds on the minimum number of transmissions required for particular instances of the ICSI problem where the corresponding side-information hypergraph can be associated with the incidence matrix of a design. This comprises Sections II-V. The remainder of the paper is concerned with upper bounds on the total transmission time for the ICCSI problem and extends the results of \cite{shanmugam2014graph} for this more general case. 
In Section II we give relevant definitions and results on incidence structures such as designs. In Section III the ICSI problem is described. In Section IV, extending results of \cite{CGC-mis-art-dau2014optimal}, we characterize those digraphs having min-rank one less than their order. In Section V we give an upper bound on the min-rank of a hypergraph whose incidence matrix can be associated with that of a 2-design and discuss a security aspect for such special instances of the ICSI problem. In Section VI we describe the ICCSI problem before finally giving several upper bounds on the transmission time of an ICCSI instance based on linear programming.

\section{Preliminaries}\label{sec:pre}

We establish some notation to be used throughout the paper. We will assume that $q$ is a power of a prime $p$, say $q=p^\ell$.
For any positive integer $n$, we let $[n]:=\{1,\dots,n\}$. We write $\fq$ to denote the finite field of order $q$ and use $\fq^{n\times t}$ to denote the vector space of all $n\times t$ matrices over $\fq$. 

Given a matrix $X \in \fq^{n\times t}$ we write ${ X}_i$ and ${ X}^j$ to denote the $i$th row and $j$th column of $X$, respectively. More generally, for subsets ${\mathcal S}\subset [n]$ and $\INSr{\mathcal T}\subset [t]$ we write ${X}_{\mathcal S}$ and 
${ X}^{\mathcal T}$ to denote the $|{\mathcal S}| \times t$ and $n \times |{\mathcal T}|$ submatrices of $X$ comprised of the rows of $X$ indexed by ${\mathcal S}$ and the columns of $X$ indexed by ${\mathcal T}$ respectively.  We write $\langle X \rangle$ to denote the row space of $X$.


A finite {\em incidence structure} $\cS=(\cP,\cB,\cI)$, consists of a pair of finite sets $\cP$ (its points) and $\cB$ (its blocks), and an incidence relation $\cI \subset \cP\times \cB$. We say that $p$ is contained in or is incident with $B$ if $(p,B) \in \cI$.

\begin{definition}
Let $t,v,k$ \INSr{and} $\gl$ be positive integers. An incidence structure $\cD=(\cP,\cB,\cI)$ is called a $t$-$(v,k,\gl)$ {\em block design} if
\begin{itemize}
\item[(1)] $|\cP|=v$;
\item[(2)] $|B|=k$ for all $B\in\cB$;
\item[(3)] every $t$-set of points of $\cP$ are contained in precisely $\gl$ blocks of $\cB$.
\end{itemize}
\end{definition}

Often a $t$-$(v,k,\gl)$ block design is simply referred to as a $t$-design.
Designs are well-studied objects in combinatorics with many applications.
The interested reader is referred to \cite{lintwilson2001,CGC-camevanl91,CGC-misc-book-ColbDin96} for further information, but we present sufficient detail here to meet our purposes.
The number of blocks $ b$ of a $t$-$(v,k,\gl)$ design is $b= \lambda \binom{v}{t}/\binom{k}{t}$ and the number of blocks containing any given point of $\cP$ is 
$r=\lambda  \binom{v-1}{t-1}/\binom{k-1}{t-1}$, \INSr{which is its} {\em replication number}. In the case of a 2-design we have $r= \lambda (v-1)/(k-1)$.  
An important parameter of a $t$-design is its {\em order}, defined to be $n=r-\lambda$.


\begin{definition}
Let $\cS=(\cP,\cB,\cI)$ be an incidence structure with $|\cP|=v$ and $|\cB|=b$. Let the points be labelled $\{p_1,\dots,p_v\}$ and the blocks be labelled $\{B_1,\dots,B_b\}$. An {\em incidence matrix}  for $\cS$ is a $b\times v$ matrix $A=(a_{i,j})$ with entries in $\{0,1\}$ such that
$$
a_{i,j}=\left\{\begin{array}{lr}
1&\mbox{if $(p_j,B_i)\in\cI$}\\
0&\mbox{if $(p_j,B_i)\notin\cI$}\end{array}\right.
$$

The {\em code} of $\cS$ over $\Fq$ is the subspace $C_q(\cS)$ of $\Fq^{|\cP|}$ spanned by the rows of $A$.
\end{definition}

\begin{definition}
Let $\cS$ be an incidence structure and let \INSr{$q$} be a prime power, the $q$-{\em rank} of $\cS$ is the dimension of the code $C_q(\cS)$ and is written
$$
\rk_q(\cS)=\dim(C_q(\cS)).
$$

\end{definition}

The following result was proved by Klemm \cite{CGC-klem86}. We will see in Section V that this gives an immediate upper bound on the min-rank
of a class of instances of the index coding problem.
\begin{theorem}\label{th:klemm}
Let $\cD=(\cP,\cB)$ be a $2$-$(v,k,\gl)$ design of order $n$ and let $p$ be a prime dividing n. Then
$$
\rk_p(\cD)\le\frac{|\cB|+1}{2}.
$$
Moreover, if $p$ does not divide $\gl$ and $p^2$ does not divide $n$, then
$$
C_p(\cD)^\perp\subseteq C_p(\cD)
$$
and $rank_p(\cD)\ge v/2$.
\end{theorem}

A $2$-$(n^2+n+1,n+1,1)$ design, for $n\ge 2$, is called a {\em projective plane} of order $n$.
A projective plane of order $n$ is an example of a {\em symmetric design}, that is, it has the same number of points as blocks, so $|\cP|=|\cB|$.


The following can be read in \cite[Theorem 6.3.1]{GCG-cd-book-assmus1992designs}.
\begin{theorem}\label{th:assm}
Let $\Pi$ be a projective plane of order $n$ and $p$ be a prime such that $p|n$. Then the $p$-ary code of $\Pi$, $\cC_p(\Pi)$, has minimum distance $n+1$. Moreover the codewords of minimal weight in $\cC_p(\Pi)$ are the scalar multiples of the rows of the incidence matrix of $\Pi$.
\end{theorem}

Chouinard, in \cite{GCG-cd-thesis-chouinard2000weight}, proved that:
 
 \begin{theorem}\label{th:pesi}
Let $C_p(\Pi)$ be a code arising from a projective plane of prime order $p$. Then no codeword has weight in the interval $[ p + 2, 2 p - 1]$.
 \end{theorem}
 

\begin{definition}
A {\em digraph} is a pair $\cG = (\cV,\cE)$ where:
\begin{itemize}
 \item $\cV$ is the set of vertices of $\cG$,
 \item $\cE \subset \cV \times \cV$ is  the set of {\em arcs} (or directed edges) of $\cG$. 
 \end{itemize}
 An arc of $\cG$ is an ordered pair $e = (u, v) \in \cE(\cG) $ for some $u, v \in \cV$. In the case that $u\ne v$, the vertex $u$ is called \INSr{the} {\em tail} of $e$ and $v$ the {\em head} of $e$. The arc $e$ is called an {\em out-going arc} of $u$ and an {\em in-coming arc} of $v$. The {\em out-degree} of a vertex $u$, $\deg_O(u)$ is the number of out-going arcs, and the {\em in-degree} of a vertex $u$, $\deg_I(u)$ is the number of in-coming arcs. 
 \INSr{$\cG$ is called an undirected graph, or a graph, if $(u,v)\in \cE$ whenever $(v,u)\in \cE$. If $\cG$ is a graph then each pair of arcs $(u,v)$ and $(v,u)$ are represented by the unordered pair $\{u,v\}$, which is called an {\em edge}. The number of vertices of a digraph is called its {\em order}.}
\end{definition}

\INSr{We assume that all digraphs have finite order.}

\begin{definition}
A {\em path} in a graph $\cG$ (respectively in a digraph), is a sequence of distinct vertices $(u_1, u_2, \dots , u_k)$, such that $\{u_i, u_{i+1}\} \in\cE$ ($(u_i, u_{i+1}) \in \cE$, respectively) for all $i \in[k - 1]$. If a path is closed, i.e. $\{u_k,u_1\} \in\cE$ ($(u_k, u_{1}) \in \cE$, respectively), then it is called {\em circuit}. \INSr{A digraph that is not a graph is called {\em acyclic} if it contains no circuits. A graph is acyclic if it has no circuits with at least 3 vertices.}
\end{definition}

Let $\nu(\cG)$ be the {\em circuit packing number} of $\cG$, namely, the maximum number of vertex-disjoint circuits in $\cG$.
A {\em feedback vertex set} of $\cG$ is a set of vertices whose removal destroys all circuits in $\cG$. Let $\gt(\cG)$ denote the {\em minimum size of a feedback vertex  set} of $\cG$.
We denote by $\ga(\cG)$ the \INSr{maximum size of vertex subset such that induced subgraph in $\cG$ is acyclic. Since such a subset of vertices is the complement of a feedback vertex set, we have $\ga(\cG)=|\cG|-\gt(\cG)$.}   
In the case that $\cG$ is a graph, $\ga(\cG)$ is the maximum size of an independent (pairwise non-adjacent) set of vertices,   

\begin{definition}
A {\em clique} of a digraph is a set of vertices that induces a complete subgraph of that digraph. A {\em clique cover} of a digraph is a set of cliques that partition its vertex set. A {\em minimum clique cover} of a digraph is a clique cover having minimum number of cliques. The number of cliques in such a minimum clique cover of a digraph is called the {\em clique cover number} of that digraph. We denote by ${\bf cc}(\cG)$ the clique cover number of a digraph $\cG$.
\end{definition}

\begin{definition}
 Let $\cG = (\cV,\cE)$ be a digraph of order $n$. A matrix $M = (m_{i,j} )\in \Fq^{n\times n}$ is said {\em to fit} $\cG$ if 
 $$
 m_{i,j}=\left\{\begin{array}{lr}
 1&\mbox{if $i=j$}\\
 0&\mbox{if $(i,j)\notin\cE$}\end{array}\right.
 $$ 
 The min-rank of $\cG$ over $\Fq$ is defined to be
 $$
\minrk_q(\cG) = \min\{\rk_q(M) : M \mbox{ fits }\cG\}
 $$
\end{definition}
 We also have analogous definitions for a graph.
 \begin{definition}
 A (directed) {\em hypergraph} $\cH$ is a pair $(\cV , \cE )$, where $\cV$ is a set of {\em vertices} and $\cE$ is a set of \INSr{{\em hyperarcs}}. A hyperarc $e$ itself is an ordered pair $(v,H)$, where $v\in \cV$ and $H\subseteq\cV$, they respectively represent the {\em tail} and the {\em head} of the hyperarc $e$.
 \end{definition}

\begin{definition}
Let $|\cV|=n$ and $|\cE|=m$. Let the hyperarcs be labelled \\$\{e_1,...,e_m\}$, a matrix $M=(m_{i,j})\in \FF_q^{m\times n}$  {\em fits} the hypergraph if
$$
m_{i,j}=\begin{cases}
		1 & \mbox{if $j$ is the tail of $e_i$}\\
		0 & \mbox{if $j$ does not lie in the head of $e_i$}\end{cases}
$$
The min-rank of $\cH$ over $\FF_q$ is defined to be
 $$
\minrk_q(\cH) = \min\{\rk_q(M) : M \mbox{ fits }\cH\}
 $$
\end{definition}

\section{Index coding with side information}\label{sec:icsi}

The Index Coding with Side Information (ICSI) problem is described as follows.
There is a unique sender $S$, who has a data matrix $X \in \FF_q^{n\times t}$.
There are also $m$ receivers, each with a request for a data packet $X_i$, and it is assumed that each receiver has some side-information, that is, a client $i$ has a subset of messages $X_{\cX_i}$, where $\cX_i\subseteq [n]$ for each $i\in[m]$. The packet requested by $i$ is denoted by $X_{f(i)}$, where $f:[m]\to[n]$ is a (surjective) {\em demand function}. Here we assume that $f(i) \notin \cX_i$ for all $i\in[m]$.
We may assume that each $i$th receiver requests only the message $X_{f(i)}$, since a receiver requesting more than one message can be split into multiple receivers, each of whom requests only one message and has the same side information set as the original \cite{GCG-cd-art-Alon2008}. 

For the remainder, let us fix $t,m,n$ to denote those parameters as described above.
Then for any $\cX=(\cX_1,\dots,\cX_n), \cX_i \subset [n]$ and map $f:[m]\to[n]$, the corresponding instance of the ICSI problem (or the ICSI instance) is denoted by 
$\cI=(\cX , f )$. It can also be conveniently described by a side-information (directed) hypergraph \cite{GCG-cd-art-Alon2008}.

\begin{definition}
Let $\cI=(\cX , f )$ be an ICSI instance. The corresponding {\em side information hypergraph} $\cH = \cH(\cX,f)$ has vertex set $\cV = [n]$ and hyperarc set $\cE$, 
defined by
$$
\cE = \{(f(i),\cX_i) : i \in[m]\}.
$$
\end{definition}

\begin{remark}
If we have $m = n$ and $f(i) = i$ for all $i \in [n]$,  the corresponding side information hypergraph has precisely $n$ hyperarcs, each with a different origin vertex. It is simpler to describe such an ICSI instance as a digraph $\cG = ( [n], \cE)$, the so-called {\em side information digraph} \cite{GCG-cd-art-Bar-Yossef2006}. For each hyperarc $(i,\cX_i)$ of $\cH$, \INSr{there} are $|\cX_i|$ arcs $(i,j)$ of $\cG$, for $j \in\cX_i$. Equivalently, $\cE = \{(i,j) : i,j \in [n], j \in\cX_i\}$.
\end{remark}

\begin{definition}
Let $N$ be a positive integer. We say that the map	
$$
E:\fq^{n\times t}\to\fq^{N},
$$
is an $\fq$-code of length $N$ for the instance $\cI=(\cX , f )$ if for each $i \in [m]$ there exists a decoding map

$$
D_i:\fq^{N}\times\FF_q^{|\cX_i|}\to\fq^t,
$$
satisfying
$$
\forall X\in\fq^{n\times t} \,:\, D_i(E({X}),X_{\cX_i})= X_{f(i)},
$$
 in which case we say that $E$ is an $\cI$-IC. 
$E$ is called an $\fq$-linear $\cI$-IC if $E(X)=LX$ for some $L\in\fq^{N \times n}$, in which case we say that $L$ represents
the code $E$. If $t=1$, $E$ is called {\em scalar} linear.
\end{definition}

The following well-known results quantify the minimal length of a linear index code in respect of its side-information hypergraph (cf. \cite{GCG-cd-art-dau2012security})
\begin{lemma}\label{lm:1}
An $\cI(\cX ,f)$-IC of length $N$ over $\Fq$ has a linear encoding map if and only if there exists a matrix $L\in \fq^{N \times n}$ such that for each $i\in[m]$, there exists a vector ${\bf u}^{(i)}\in \Fq^n$ satisfying
\begin{eqnarray}
\supp( {\bf u}^{(i)})\subseteq\cX_i\\
{\bf u}^{(i)}+{\bf e}_{f(i)}\in \langle L\rangle.
\end{eqnarray}
\end{lemma}

\begin{theorem}\label{cor:min-rank}
Let $\cI=(\cX,f)$ be an instance of the ICSI problem, and $\cH$ its hypergraph. Then the optimal length of a $q$-ary linear $\cI$-IC is $\minrk_q(\cH)$.
\end{theorem}


Achievable schemes based on graph-theoretic models for constructing index codes (i.e. upper bounds for index coding) were largely studied  \cite{GCG-cd-art-Alon2008,GCG-cd-art-Bar-Yossef2006,6034005,shanmugam2014graph}.

One of these methods comes from the well-known fact that all the users forming a clique in the side information digraph can be simultaneously satisfied by transmitting the sum of their packets \cite{662940}. This idea shows that the number of cliques required to cover all the vertices of the graph (the clique cover number) is an achievable upper bound.

A lower bound on the min-rank of a digraph was given in \cite{GCG-cd-art-Bar-Yossef2006}. An acyclic digraph has min-rank equal to its order (see for instance \cite{GCG-cd-art-Bar-Yossef2006}) and for any subgraph $\cG'$ of a graph $\cG$ we have 
$$
\minrk_q(\cG')\le\minrk_q(\cG).
$$
Let $M$ be a matrix that fits $\cG$, the sub-matrix $M'$ of $M$ restricted on the rows and columns indexed by the vertices in $\cV(\cG')$ is a matrix that fits $\cG'$. These two results are summarized in the following theorem.

\begin{theorem}\label{th:cap}
Let ${\cG}$ be a digraph. Then
$$
\ga(\cG)\le\minrk_q(\cG)\le\mathbf{cc}({\cG}).
$$ 
\end{theorem}

Instead of covering with cliques, one can cover the vertices with circuits. In \cite{6034005} the {\em circuit-packing bound} was implicitly introduced by the authors. Indeed, Chaudhry and Sprintson construct a linear index code partitioning the graph of the ICSI instance in disjoint circuits. The same bound was explicitly given in the work of Dau \emph{et al.} \cite{CGC-mis-art-dau2014optimal}. It is based on the observation that the existence of a circuit of length $k$ in the side-information digraph $\cG$ 
requires at most $k-1$ transmissions to satisfy the demands of the corresponding $k$ users. Therefore a collection of $\nu$ vertex disjoint circuits corresponds to a `saving' of at least \INSe{$\nu$ transmissions}. The bound is stated as follows:
Let $\nu(\cG)$ be the circuit-packing number of a graph $\cG$ of order $n$. Then
$$
\minrk_q(\cG)\le n-\nu(\cG).
$$

In \cite{tehrani2012bipartite} the following result is given, leading the authors to introduce the {\em partition multicast scheme}, which outperforms the circuit-packing number.

\begin{proposition}\label{prop:bound}
Let $\cG$ be a graph of order $n$. Then 
$$
\minrk_q(\cG)\le n-\min_{v \in \cV}\deg_O(v),
$$
for any $q>n$.
\end{proposition}

The broadcast rate of an IC-instance $\cI$ \cite{GCG-cd-art-Alon2008} is defined as follows, with respect to a prime $p$. 
\begin{definition}
Let $\cI=(\cX,f)$ be an IC instance. We denote by $\gb_t(\cI)$ the minimal number of symbols required to broadcast the information to all receivers, when the block length is $t$, over all possible extensions of $\FF_p$, i.e.
$$
\gb_t(\cI)=\inf_q\{N\mid \text{$\exists$ a $q$-ary index code of length $N$ for $\cI$}\}.
$$
Moreover we denote by $\gb(\cI)$ the limit
$$
\gb(\cI)=\lim_{t\rightarrow \infty}\frac{\gb_t(\cI)}{t}=\inf_{t}\frac{\gb_t(\cI)}{t}.
$$
\end{definition}

\INSe{In the following, we will also use the notation $\gb(\cG)$ to indicate the broadcast rate of any instance that has $\cG$ as side-information graph.}

The graph parameter $\minrk_q(\cG)$ completely characterizes the length of an optimal linear index code. Bar-Yossef et al. \cite{GCG-cd-art-Bar-Yossef2006,GCG-cd-art-Bar-Yossef2011} showed that in various cases linear codes attain the optimal word length, and they conjectured that the minimum broadcast rate of a graph $\cG$ was $\minrk_2(\cG)$ also for non-linear codes. Lubetzky and Stav in \cite{lubetzky2009nonlinear} disproved this conjecture.

%
%
%

In the works of Alon \emph{et al.} \cite{GCG-cd-art-Alon2008} and Shanmugam \emph{et al.} \cite{shanmugam2013local}, it was shown that results based on partitioning the vertices of a graph $\cG$ in cliques lead to a family of stronger bounds on $\gb(\cG)$, starting with an LP relaxation called \emph{fractional chromatic number} \cite{GCG-cd-art-Alon2008} and the stronger \emph{fractional local chromatic number} \cite{shanmugam2013local}. In \cite{shanmugam2014graph} the authors extended all these schemes to the case of hypergraphs. 

\section{On directed graphs with min-rank one less than the order}

In the work of Dau \emph{et al.} \cite{CGC-mis-art-dau2014optimal} the authors characterize the undirected graphs of order $n$ having min-rank $n-1$. Here we extend this result to include directed graphs over a sufficiently large field. \INSe{Our result relies in part on the following lemma, which is a construction of a digraph $\cG'$ of minrank one less that a digraph $\cG$, obtained from $\cG$ by contracting an arc.}


\begin{lemma}\label{lm:reduction}
Let $\cG=(\cV,\cE)$ be a directed graph of order $n$ such that there exist $i_1,i_2\in\cV$ with 
\begin{itemize}
\item[$(1)$] $(i_1,i_2)\in\cE$ and $(i_2,i_1)\notin\cE$
\item[$(2)$] $\deg_O(i_1)=1$.

\end{itemize}
Let $\cG'=(\cV',\cE')$ with $\cV'=\cV\setminus\{i_1\}$ and \\
$\cE'=\left(\cE\cup\{(j,i_2)\,|\,(j,i_1)\in \cE\}\right)\setminus(\{(i_1,i_2)\}\cup\{(j,i_1)\,|\,(j,i_1)\in \cE\})$. Then
$$
\minrk_q(\cG)=\minrk_q(\cG')+1
$$
for any $q$.
\end{lemma}
 
\begin{proof}
Let $M=(m_{i,j})$ be a matrix that fits $\cG$ of minimum rank. We may assume that $i_1=1$ and $i_2=2$ so that the first two rows of $M$ are 

$$
M_1=(1,\ga,0,\dots,0)
$$
and
$$
M_2=(0,1,m_{2,3},\dots,m_{2,n}).
$$
If $\ga=0$ then it is easy to check that deleting the first row and the first column of $M$ we obtain $M'$ of rank $\rk
(M)-1$ that fits $\cG'$.

Now suppose that $\ga\ne 0$. We may assume that the rows $M_1,M_2,\dots,M_{\minrk_q(\cG)}$ are linearly independent.

For each vertex $i \in \cV\setminus\{1\}$, label the corresponding vertex in $\cV'$ by $i-1$. 
Then construct the $(n-1) \times (n-1)$ matrix $M'$ whose $i$-th row is obtained from the $i+1$-th row of $M$ in the following way: for $i=1,\dots, \minrk_q(\cG)-1$
let
$$
M'_i=(m_{i+1,1}+m_{i+1,2},m_{i+1,3},\dots,m_{i+1,n}),
$$
and for $i=\minrk_q(\cG),\dots,n-1$ we define

$$
M'_{i}=(m_{i+1,1}+m_{i+1,2}-\gl_1(1+\ga),m_{i+1,3},\dots,m_{i+1,n})
$$
where $\gl_1 \in \fq$ satisfies $M_{i+1}=\sum_{r=1}^{\minrk_q(\cG)} \gl_{r} M_r$ for some $\gl_r$. The matrix $M'$ fits $\cG'$, so
$$
\minrk_q(\cG')\le \rk
(M')\le\minrk_q(\cG)-1.
$$

Conversely, let $M'=(m'_{i,j})$ be a matrix that fits $\cG'$ having rank $\minrk_q(\cG')$ and suppose the rows $M'_1,M'_2,\dots,M'_{\minrk_q(\cG')}$ are linearly independent. Let $I=\{j\,|\,(j,1)\in \cE\}$ be the set of vertices of $\cG$ with outgoing arcs directed to $1$.  We construct the matrix $M$ such that
$$
M_1=(1,-1,0,\dots,0),
$$
$$
M_i=(m'_{i-1,1},0,m'_{i-1,2},\dots,m'_{i-1,n-1}),
$$
for $i\in I \cap \{2,\dots,\minrk_q(\cG')+1\}$ and 
$$
M_i=(0,m'_{i-1,1},m'_{i-1,2},\dots,m'_{i-1,n-1}),
$$
for $i\in ([n] \backslash I) \cap \{2,\dots,\minrk_q(\cG')+1\}$. 
For $i>\minrk_q(\cG')+1$ we have that the $i-1$-th row of $M'$ is given by\\
\vskip 0.1cm
$$
M'_{i-1}=\sum_{r=1}^{\minrk_q(\cG')} \gl_r M'_r,
$$
for some $\gl_r\in\fq$. If $i \in I$, we put
$$
M_i=\left(m'_{i-1,1},0,m'_{i-1,2},\dots,m'_{i-1,n-1}\right)
$$
and hence obtain
$$
M_i=\gl M_1+\sum_{r=2}^{\minrk_q(\cG')+1} \gl_{r-1} M_r
$$
where the $\gl_r$ are the coefficients in the linear combination of $M'_{i-1}$, with respect to the first $\minrk_q(\cG')$ rows of $M'$, and $\gl=\sum_{r\notin I}\gl_{r-1}\INSe{m'_{r-1,1}}$.
If $i\notin I$ we set 
$$
M_i=\left(0,m'_{i-1,1},m'_{i-1,2},\dots,m'_{i-1,n-1}\right)
$$
and we have
$$
M_i=\gl M_1+\sum_{r=2}^{\minrk_q(\cG')+1} \gl_{r-1} M_r
$$
where $\gl=-\sum_{r\in I}\gl_{r-1}\INSe{m'_{r-1,1}}$.

Then $M$ fits $\cG$ and

$$
\minrk_q(\cG)\le \rk(M)\le\minrk_q(\cG')+1.
$$
~\end{proof}

Note that the digraph $\cG'$ of Lemma \ref{lm:reduction} is the contraction of the digraph $\cG$ along the arc $(i_1,i_2)$.

\begin{example}
Let $\cG$ and $\cG'$ be the two digraphs shown in Figure \ref{fig:contrazione}. The nodes $1$ and $2$ of $\cG$ satisfy the conditions \INSe{of Lemma \ref{lm:reduction}},
so we can reduce $\cG$ to $\cG'$. Consider the matrix
$$
M=\left[\begin{array}{cccc}
1&-1&0&0\\
0&1&1&1\\
1&0&1&1\\
1&0&1&1\end{array}
\right],
$$
which fits $\cG$. We have $M_3=M_4=M_1+M_2$, constructing $M'$ as in the lemma above we obtain

$$
M'=\left[\begin{array}{ccc}
1&1&1\\
1&1&1\\
1&1&1\end{array}
\right].
$$
$M'$ fits $\cG'$. Conversely, from $M'$ we obtain $M$, and $\rk
(M)=\rk
(M')+1$.
\begin{figure}[!ht]\label{fig:contrazione}
\centering
\subfigure[$\cG$]{
\begin{tikzpicture}[->,>=stealth',shorten >=0.5pt,auto,node distance=1.5cm,
  thick,main node/.style={circle,fill=black!20,draw,font=\sffamily\bfseries}]

  \node[main node] (1)  {3};
  \node[main node] (4) [below of=1] {4};
  \node[main node] (2) [above right of=1] {1};
  \node[main node] (3) [above left of=1] {2};
  \path[every node/.style={font=\sffamily\small}]
    (1) edge node [left] {}(2)
        edge [bend right] node[left] {}(4)
       
    (2) edge [bend right] node [right]{} (3)

    (3) edge node [right]{} (1)
        edge [bend right] node [right]{} (4)
        
    (4) edge [bend right] node [right]{} (1)
    edge [bend right] node [right]{} (2)
        
        ;

\end{tikzpicture}
}
\quad\quad
\subfigure[$\cG'$]{

\begin{tikzpicture}[->,>=stealth',shorten >=0.5pt,auto,node distance=1.5cm,
  thick,main node/.style={circle,fill=black!20,draw,font=\sffamily\bfseries}]

  \node[main node] (1)  {2};
  \node[main node] (2) [ right of=1] {1};
  \node[main node] (3) [left of=1] {3};
  \path[every node/.style={font=\sffamily\small}]
    (1) edge [bend right] node [left] {}(2)
           edge [bend right] node[left] {}(3)
       
    (2) edge [bend right] node [right]{} (3)
        edge [bend right] node [right]{} (1)
        
    (3) edge [bend right] node [right]{} (1)
        edge [bend right] node [right]{} (2)
        
     ;

\end{tikzpicture}
}
\caption{Contraction graph}
\end{figure}
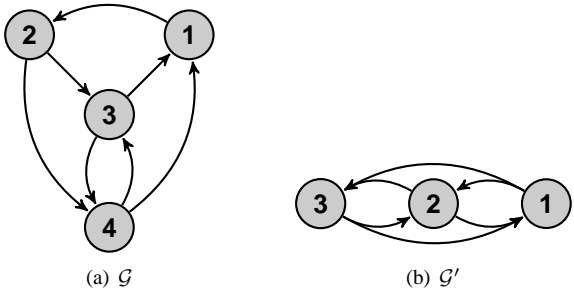

\end{example}

\begin{lemma}\label{lm:n-2}
Let $\cG$ be a directed graph of order $n$ such that $\gt(\cG)=2$. Then $\minrk_q(\cG)=n-2$, for any $q>n$.
\end{lemma}
\begin{proof}
As observed in Theorem \ref{th:cap}, $n-\tau(\cG) \leq \minrk_q(\cG)$, so we need only to prove that $\minrk_q(\cG)\le n-2$.

We may suppose without loss of generality that there does not exist $i\in \cV$ with out-degree less than $1$, otherwise, from Lemma \ref{lm:reduction} we can delete the node $i$ and consider the induced subgraph $\cG'$, which satisfies $\minrk_q(\cG')=\minrk_q(\cG)-1$.

Since $\gt(\cG)=2$, we have $\nu(\cG) \in \{1,2\}$. \INSe{Since $\minrk_q(\cG) \leq n- \nu(\cG)$, if $\nu(\cG)=2$ then we have our claim immediately. Assume then that $\nu(\cG)=1$. We apply Lemma \ref{lm:reduction}, iteratively. Note that each time we reduce a graph $\cG$ by an appropriate arc contraction, we obtain $\cG'$ with $\gt(\cG')=2$ and $\nu(\cG')=1$. Moreover, for each contraction of an arc of the graph}, we only shorten the circuits that pass through the node that we delete, and we do not create any new circuit from the fact that the out-degree of the node is $1$.

At the point that Lemma \ref{lm:reduction} is no longer applicable, there are two possible cases:
\begin{enumerate}
\item[1)] the out-degree of each node of the reduced graph $\cG'$ is at least $2$,
\item[2)]  there exists $i_1$ with out-degree $1$ and $(i_1,i_2),(i_2,i_1)\in \cE'$.
\end{enumerate}
This last case is not possible, in fact if we consider the circuit $C=(i_1,i_2)$, from $\gt(\cG')=2$ we have that there exists a circuit $C'$  which remains after deleting $i_2$. Then, $C'$ does not pass through $i_1$ otherwise it has to pass through $i_2$. Then $C$ and $C'$ are 
 disjoint, but this is not possible because $\nu(\cG')=1$.

Therefore, reducing $\cG$ we obtain $\cG'$ with $k$ fewer nodes and all nodes have out-degree at least $2$. Then from Proposition \ref{prop:bound} and Lemma \ref{lm:reduction} it follows that
$$
\minrk_q(\cG)=\minrk_q(\cG')+k \le n-2.
$$
~\end{proof}

\begin{corollary}
Let $\cG$ be a directed graph of order $n$ such that $\gt(\cG)=2$. Then for any $q>n$, $\minrk_q(\cG)=\gb(\cG)$.
\end{corollary}

We have now our main result of this section. 
\begin{corollary}\label{th:bound}
Let $\cG$ a graph of order $n$ and let $q>n$. Then $\minrk_q(\cG)= n-1$ if and only if $\gt(\cG)=1$. Moreover in that case we have $\gb(\cG)=n-1$ if and only if $\gt(\cG)=1$.
\end{corollary}
\begin{proof}
If $\gt(\cG)=1$ then $\nu(\cG)=1$ and we have $\minrk_q(\cG)= n-1$.

Conversely towards a contradiction assume that $\gt(\cG)\ge 2$. Then consider a subgraph $\cG'$ of $\cG$ with $\gt(\cG)=2$. From Lemma \ref{lm:n-2} we have our claim.
\end{proof}

This last theorem implies that the problem of deciding whether or not a digraph has min-rank $n-1$, over a sufficiently large field, can be solved in polynomial time, using a depth-first search algorithm (see for instance \cite{Cormen:2001:IA:580470}) that verifies in a polynomial time whether or not a graph is acyclic.

\begin{corollary}
Let $\cG$ be a digraph of order $n$ and $q>n$. Then deciding whether $\minrk_q(\cG) =n-1$ can be done in polynomial time ($\cO(n^3)$).
\end{corollary}

\INSe{\begin{remark}
		In the final stages of the writing of this paper we learned of Ong's result \cite{Ong14}.
		In fact Lemma \ref{lm:n-2}, (although obtained independently) and its immediate corollary follows from \cite[Theorem 1]{Ong14}, which is a stronger result, since it holds without any restrictions on $q$. That is, 
		\begin{theorem}[\cite{Ong14}]\label{th:Ong}
			Let $\cG$ be a directed graph of order $n$ satisfying $\gt(\cG)\leq 2$. Then 
			$$\minrk_q(\cG)=\gb(\cG) = n-\gt(\cG).$$
		\end{theorem}	  		
	    The proof of Theorem \ref{th:Ong} relies on showing that $\cG$ contains a particular subgraph $\cG_{{\rm sub}}$ and then devising a coding scheme for $\cG$ based on the existence of $\cG_{{\rm sub}}$. The proof given in \cite{Ong14} is a non-trivial graph-theoretic proof and goes through a careful case-by-case analysis. The proof of Lemma \ref{lm:n-2} given here is rather more straightforward, being based on the construction of a new graph $\cG'$ obtained by iterative contractions of the original graph $\cG$, following from Lemma \ref{lm:reduction}. Such a result could be helpful also to decrease the size of a graph and thus to optimize the computation of the min-rank of the graph.
	    The hypothesis that $q>n$ follows since we invoke the partition multicast solution (Proposition \ref{prop:bound}), therefore requiring the existence of a maximum distance separable code. 
\end{remark}}

 In the following table we report the values of the min-rank for graphs and directed graphs with near-extreme min-rank (i.e. $1,2,n-2,n-1$ and $n$).

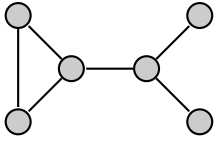
\begin{figure}[!h]\label{fig:graph}
\centering
\begin{tikzpicture}[scale=0.85,-,>=stealth',shorten >=0.5pt,auto,node distance=1cm,
  thick,main node/.style={circle,fill=black!20,draw,font=\sffamily\bfseries}]

  \node[main node] (1)  {};
   \node[main node] (3) [left of=1] {};
  \node[main node] (5) [below left of=3] {};
   \node[main node] (6) [above left of=3] {};
 \node[main node] (4) [below right of=1] {};
  \node[main node] (2) [above right of=1] {};
 
  \path[every node/.style={font=\sffamily\small}]
    (1) edge [] node [] {}(4)
           edge [] node [] {}(3)
           edge [] node [] {}(2)
       
    (3) edge [] node [] {}(5)
           edge [] node [] {}(6)        
    (6) edge [] node [] {}(5)
        
     ;

\end{tikzpicture}
\caption{Forbidden subgraph}
\end{figure}
\begin{center}
\begin{tabular}[h]{| p{1.6cm} | p{3cm} | p{3cm}| }
  \hline
  {\bf Minrank} & {\bf Graph} ${\bf \cG}$  &{\bf Digraph $\cD$} \\
  \hline
  1 & $\cG$ is complete (trivial) & $\cD$ is complete (trivial)\\
    \hline
  2 & $\bar\cG$ is $2$ colorable \cite{CGC-mis-art-peeters1996orthogonal} &for $q=2$, if $\bar\cD$ is $3$-fair colorable \cite{CGC-mis-art-dau2014optimal}\\
    \hline
 $ n-2$ & $\cG$ has maximum matching $2$ and does not contain the graph in Figure 2 
              \cite{CGC-mis-art-dau2014optimal} & unknown \\
  \hline
  $n-1$ & $\cG$ is a star graph \cite{CGC-mis-art-dau2014optimal}& for $q>n$, $\gt(\cD)=1$  \INSr{Corollary} \ref{th:bound} \\
        &                                                        & \INSe{for any $q$, $\gt(\cD)=1$ Theorem \ref{th:Ong}}\\
  \hline
 $ n $ & $\cG$ has no edges (trivial) & $\cD$ is acyclic (trivial) \cite{GCG-cd-art-Bar-Yossef2006}\\
  \hline
\end{tabular}
\end{center}

\section{A bound from t-designs}\label{sec:design}

In this section we study the case for which an incidence structure, in particular a $2$-$(r^2+r+1,r+1,1)$ or projective plane, arises from the side information. \INSr{This yields an immediate upper bound on the min-rank of the hypergraph, based on known results on the ranks of incidence matrices. Furthermore, we show that secrecy and privacy are attainable for such configurations. Towards secrecy, we show that if an instance fits a projective plane, then a receiver may recover only its requested data, and no more. On the matter of privacy, we identify a constraint on the side information of an adversary hearing the broadcast such that it cannot access the receivers' requested data.}
We may assume without loss of generality that $t=1$.

\begin{definition}
We said that an instance, $\cI=(\cX , f )$, of the ICSI problem {\em contains an incidence structure} $\cS=(\cP,\cB)$ if 
\begin{itemize}
\item[$1)$] $\cP=[n]$ and $|\cB|\le m$;
\item[$2)$] for each $i\in[m]$ there exists $B\in \cB$ such that $f(i)\in B$ and $B\setminus\{f(i)\}\subseteq\cX_i$.
\end{itemize}
Moreover we said that the instance {\em coincides} with the incidence structure $\cS$ if the following condition is satisfied.
\begin{itemize}
\item[$2')$] for each $i\in[m]$ there exists $B\in \cB$ such that $f(i)\in B$ and $B\setminus\{f(i)\}=\cX_i$.
\end{itemize}
\end{definition}

 We immediately obtain the following proposition.
 
\begin{proposition}\label{th:1}
Let $\cI=(\cX , f )$ be an instance of ICSI problem and $\cH$ let be the corresponding hypergraph. \INSr{If the instance contains a $2$-$(n,k,\gl)$ design $\cD=(\cP,\cB)$ then for all $q$ a power of a prime $p$ such that $p$ divides the order of $\cD$ it holds that}
$$
\minrk_q(\INSe{\cH})\le \frac{m+1}{2}.
$$

\end{proposition}
\begin{proof}
Let $D$ be the incidence matrix of $\cD$. Then for the Theorem \ref{th:klemm} we have that the $p$-rank of $\cD$ is less or equal to $\frac{m+1}{2}$.

Now, it is easy to check that $D$ fits $\INSe{\cH}$, so
$$
\minrk_q(\INSe{\cH})\le \rk_q(\INSe{\cD})\le \rk_p(\cD)
$$

and that concludes the proof.
\end{proof}

\begin{remark}
To compute the min-rank of a hypergraph is an NP-hard problem {\cite{CGC-mis-art-peeters1996orthogonal}}, however, if there exists a $2$-design as in Proposition \ref{th:1} it is possible to have a bound on this value and we can use the linearly independent rows of its incidence matrix to decrease the number of transmissions.
\INSr{We remark further that this result does not require $q$ to be large, and shows the existence of a class of instances with transmission rate much less than predicted by other bounds. For example, it is known that if an instance fits the incidence matrix of a projective plane of order $r$ and $q>r^2+r+1$ then $\minrk_q({\cH}) \leq r^2+r+1 - r=r^2+1$ (see, for example \cite{byrne2015error}), which is significantly greater than the bound $\minrk_q({\cH}) \leq (r^2+r+2)/2$, given by Proposition \ref{th:1}}. 
\end{remark}

\begin{example}\label{ex:fano}
Consider the instance of the ICSI problem $\cI$ given by $n=m=7$, and  $f(i)=i$ for $i=1,\dots,7$. Let the side information be
$$
\cX_1=\{2,3\}, \,\cX_2=\{6,7\}, \, \cX_3=\{5,7\},\,\cX_4=\{2,5\},
$$
$$
\cX_5=\{1,6\},\,\cX_6=\{3,4\},\,\cX_7=\{1,4\}.
$$

Consider the blocks
$$
B_1=\{1,2,3\}, \,B_2=\{2,6,7\}, B_3=\{3,5,7\},\,B_4=\{2,4,5\},
$$
$$
B_5=\{1,5,6\},\,B_6=\{3,4,6\},\,B_7=\{1,4,7\}.
$$

These blocks form the Fano plane as in Figure 3. 
This is a $2$-$(7,3,1)$ design of order $2$ and the design is contained in the side information. The $2$-rank of the design is $4$. Then we can consider $4$ linearly independent rows of the incidence matrix of the Fano plane, and encode the message using those reducing the number of transmissions from $7$ to $4$.

\INSr{It can be checked that distribution of the ranks of the matrices that fit this incidence is given by}
	$$ \INSr{(4,1), (5,238), (6,6575), (7,9570)},$$
	thus the bound is sharply met in this instance.  
\INSr{Moreover, an optimal encoding matrix $L$ for this instance must have row space spanned be the rows of this incidence matrix; there is a unique optimal solution, up to left multiplication by an invertible matrix.}

\begin{figure}[!ht]
\centering
\includegraphics[scale=0.2]{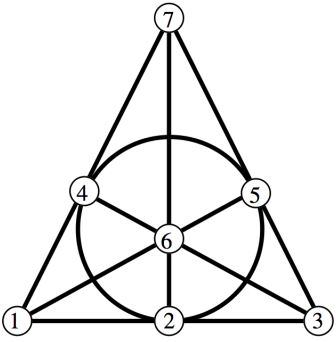}
\caption{Fano plane}
\end{figure}\label{fig:fano}

\end{example}


Now we consider the case when an instance $\cI=(\cX , f )$ of the ICSI problem contains a $2$-$(r^2+r+1,r+1,1)$ design, and the matrix corresponding to the index code is composed of the linearly independent rows of the incidence matrix of the design.
We recall that a $2$-$(r^2+r+1,r+1,1)$ design has order $r$ and the code of the design over $\FF_p$, with $p$ a prime divisor of $r$, has minimum distance equal to $r+1$ (Theorem \ref{th:assm}).

\begin{theorem}
If the instance $\cI$ of the ICSI problem coincides with the $2$-$(r^2+r+1,r+1,1)$ design, then no receiver $i\in[m]$ can recover a message $X_j$ with $j\notin\cX_i\cup\{f(i)\}$.
\end{theorem}
\begin{proof}
Let $\cD$ be the $2$-$(r^2+r+1,r+1,1)$ design. Suppose that $R_i$ wants to recover $X_j$ with $j\notin\cX_i\cup\{f(i)\}$. From Lemma \ref{lm:1} it is able to do so if and only if there exists a vector ${\bf u}\in\FF_p^{n}$, $n=r^2+r+1$, such that $\supp({\bf u})\subseteq\cX_i\cup\{f(i)\}$ and ${\bf u}+{\bf e}_j\in \INSe{C_p}(\cD)$.
If this vector is a codeword of the code, at least $r+1$ positions are different from $0$. 
Now consider the vector ${\bf 1}_{\cX_i}+{\bf e}_{f(i)}\in \INSe{C_p}(\cD)$, where ${\bf 1}_{\cX_i}$ is the vector in $\FF_p^n$ with $1$'s in the positions contained in $\cX_i$. We have $|\supp({\bf u}+{\bf e}_j)\cap \supp({\bf 1}_{\cX_i}+{\bf e}_{f(i)})|\ge r$ and also there are at least $2$ positions of ${\bf u}+{\bf e}_j$ in this intersection that have the same value (we can use only the $p-1$ values of ${\FF_p}\setminus\{0\}$ for these $r$ positions). Suppose that this value is $\ga\in\FF_p\setminus\{0\}$, then we have $d({\bf u}+{\bf e}_j,\ga({\bf 1}_{\cX_i}+{\bf e}_{f(i)}))\le r$. So ${\bf u+e_j}$ is not a codeword of $\INSe{C_p}(\cD)$, which means that $R_i$ is not able to recover $X_j$.
\end{proof}

Encoding with a matrix whose rowspace contains the blocks of a projective plane guarantees the secrecy of the transmission.

Assume, now, the presence of an adversary $A$ who can listen to all transmissions. The adversary is assumed to possess side information $\{X_h\,|\,h\in\cX_A\subseteq [n]\}$. In \cite{GCG-cd-art-dau2012security}, it is shown that for a transmission matrix $L$ for a linear index code representing $\cI=(\cX,f)$, if $|\cX_A|\le d-2$, where $d$ is the minimum distance of the code $\langle L\rangle$, then $A$ is not able to recover an element $X_j$ with $j\notin\cX_A$. 

Consider now an instance $\cI=(\cX , f )$ of the ICSI problem containing a $2$-$(p^2+p+1,p+1,1)$ design, where $p$ is a prime number. Suppose the 
matrix $L$ as above is used as an encoding matrix. Then we obtain the following result.

\begin{theorem}
If $|\cX_A| \le 2p-2$ and for each block $B$ of the design $|\cX_A\cap B|\le p-1$, then $A$ is not able to recover $X_j$ for any $j\notin\cX_A$.

\end{theorem}
\begin{proof}
If $p$ is even, then the result follows from the fact that $|\cX_A| \le 1 = d-2$.
Let $p$ be odd. We know from Theorem \ref{th:pesi} that in the code generated by the incidence matrix of a $2$-$(p^2+p+1,p+1,1)$ design there are no codewords with weights in $[p+2,2p-1]$. To recover the message $X_j$, $A$ needs a codeword of weight $p+1$. Such codewords are those corresponding to some block $B$, that is a vector of the form
$$
\sum_{i\in B}{\bf e}_i
$$
and its scalar multiples.

So $A$ recovers $X_j$ if and only if there exists ${\bf u}+{\bf e}_j\in C$ with $\supp({\bf u})\subset\cX_A$ and $|\supp({\bf u})|=p$. Here $C$ means the code of the projective space.
Then $\supp({\bf u}+{\bf e}_j)=B$ for some block $B$, and so $|(\cX_A\cup\{j\})\cap B|\geq p+1$.
\end{proof}

\section{Index Coding with Coded Side Information}\label{sec:iccsi}

In \cite{GCG-cd-art-shum2012broadcasting} the authors generalized the index coding problem so that coded packets of a data matrix $X$ may be broadcast or part of a user's cache. 
This finds applications, for example, in broadcast channels with helper relay nodes. 
We present the model with coded side information in the following section.

\subsection{Preliminaries on the ICCSI Problem}

As before there is a data matrix $X \in \FF_q^{n\times t}$ and a set of $m$ receivers or users.
For each $i \in[m]$, the $i$th user seeks some linear combination of $X$, say $R_i X$ for some $R_i \in \FF_q^n$. 
A user's cache comprises a pair of matrices 
$$V^{(i)} \in \FF_q^{d_i \times n} \text{ and }\Lambda^{(i)}\in \FF_q^{d_i\times t}$$
related by the equation
$$\Lambda^{(i)} = V^{(i)}X.$$ 
While $X$ is unknown to user $i$, it is assumed that any vector $v$ in the row spaces of $V^{(i)}$ and the respective $\lambda=v X$ can be generated at the $i$th receiver.
We denote these respective row spaces by $\cX^{(i)}:=\langle V^{(i)} \rangle  $ and ${\mathcal L}^{(i)}:=\{v\cdot X\mid v\in \cX^{(i)}\}$ for each $i$.
The side information of the $i$th user is $(\cX^{(i)},{\mathcal L}^{(i)})$.
Similarly, the sender has the pair of row spaces $(\cX^{(S)}, {\mathcal L}^{(S)})$ for matrices 
$$V^{(S)} \in \FF_q^{d_S \times n} \text{ and } \Lambda^{(S)} = V^{(S)}X \in \FF_q^{d_S}$$ 
and does not necessarily possess $X$ itself. 

The $i$th user requests a coded packet $R_iX \in {\mathcal L}^{(S)}$ with $R_i \in \cX^{(S)} \backslash \cX^{(i)}$. 
We denote by $R$ the $m \times n$ matrix over $\fq$ with each $i$th row equal to $R_i$. The matrix $R$ thus represents the requests of all $m$ users.
We denote by 
$$ \cX: = \{ A \in \fq^{m \times n} : A_i \in \cX^{(i)}, i \in [m]\},$$
so that $\cX=\oplus_{i \in [m]} \cX^{(i)}$ 
is the direct sum of the $\cX^{(i)} $ 
as a vector space over $\fq$. 
Similarly, we write $\INSe{\oplus \cX^{(S)}}: =\oplus_{i \in [m]} \cX^{(S)} = \{Z \in \fq^{m \times n} : Z_i \in \cX^{(S)}\}$.

For the remainder, we let $\cX,\cX^{(S)},\INSe{\oplus \cX^{(S)}} ,R$ be as defined above and write $\cI=(\cX,\cX^{(S)},R)$ to denote an instance of the ICCSI problem for these parameters. As before, for the ICCSI instance $\beta_t(\cI)$ denotes the minimum broadcast rate for block-length $t$ where the encoding is over all possible extensions of $\FF_p$. That is, for $\cI=(\cX,\cX^{(S)},R)$
$$
\gb_t(\cI)=\inf_q \{N\mid \text{$\exists$ a $q$-ary index code of length $N$ for $\cI$}\}.
$$
\INSe{The optimal broadcast rate is given by the limit
$$
\gb(\cI)=\lim_{t\rightarrow \infty}\frac{\gb_t(\cI)}{t}=\inf_{t}\frac{\gb_t(\cI)}{t}.
$$}
\begin{definition}
Let $N$ be a positive integer. We say that the map	
$$
E:\fq^{n\times t}\to\fq^{N},
$$
is an $\fq$-code for $\cI=(\cX,\cX^{(S)},R)$ of length $N$ if for each $i$th receiver, $i \in [m]$ there exists a decoding map

$$
D_i:\fq^{N}\times\cX^{(i)}\to\fq^t,
$$
satisfying
$$
\forall X\in\fq^{n\times t} \,:\, D_i(E({X}),A)=R_i X,
$$
for some vector $A \in \cX^{(i)}$, in which case we say that $E$ is an $\cI$-IC. 
$E$ is called an $\fq$-linear $\cI$-IC if $E(X)=LV^{(S)}X$ for some $L\in\fq^{N \times d_S}$, in which case we say that $L$ represents
the code $E$. 
\end{definition}

Given an instance $\cI =(\cX,\cX^{(S)},R)$ and a matrix $L\in \fq^{N\times d_S}$ that represents an $\cI$-IC, we write $\cL$ to denote the space $\langle L V^{(S)} \rangle.$

We have the following (see \cite{byrne2015error,GCG-cd-art-shum2012broadcasting}).

\begin{lemma}\label{lemdecode}
   Let $L\in \fq^{N \times d_S}$. Then $L$ represents a $\fq$-linear $\cI$-IC index code of length $N$ if and only if for each $i \in [m]$,
   $R_i \in \cL + \cX^{(i)}. $
\end{lemma}
\begin{remark}\label{rm:comb}

If the equivalent conditions of the above lemma hold we have that for each $i \in [m]$, 
$R_i=\bb^{(i)} L V^{(S)}+ \ba^{(i)} V^{(i)}$ for some vectors $\ba^{(i)},\bb^{(i)}$. 
So User $i$ decodes its request by computing 
$$R_iX = \bb^{(i)}L V^{(S)}X+\ba^{(i)}V^{(i)}X=\bb^{(i)}Y+\ba^{(i)}\Lambda^{(i)},$$ where $Y$ is the received message.
\end{remark}

\begin{remark}
	The ICSI problem as introduced before is indeed a special case of the ICCSI problem. Setting $V^{(S)}$ to be the $n \times n$ identity matrix, $R_i = {\bf e}_{f(i)}$ and $V^{(i)}$ to be the $d_i \times n$ matrix with rows $V^{(i)}_j = {\bf e}_{i_j}$ for each $i_j \in \cX_i$ yields $\cX^{(i)} = \langle {\bf e}_{j} : j \in \cX_i \rangle$, so that $\supp (\bv) \subset \cX_i$ if and only if $\bv \in \cX^{(i)}$.  
\end{remark}

The analogue of the min-rank is as follows:

\begin{definition}[\cite{byrne2015error}]
The min-rank of the instance $\cI=(\cX,\cX^{(S)},R)$ of the ICCSI problem over $\fq$ is
$$
\gk(\cI) = \min\left\{\begin{array}{cc}
\multirow{2}{*}{\rk(A + R)\,:}& A \in \fq^{m \times n},\\
&  A_i \in \cX^{(i)} \cap \cX^{(S)}, \forall i \in [m] \end{array}\right\} .           
$$
\end{definition}
%

Note that $\gk(\cI)$ measures the rank distance of the $m \times n$ matrix $R$ to the $\fq$-linear matrix code $\cX \cap (\oplus \cX^{(S)})$.

As in the ICSI case, the length of an optimal $\fq$-linear ICCSI index code is characterized by the min-rank of the instance.
 
\begin{lemma}[\cite{byrne2015error}]
The length of an optimal $\fq$-linear index code for $\cI{\hskip -0.1cm} ={\hskip -0.1cm}(\cX, {\hskip -0.01cm}\cX^{(S)}{\hskip -0.01cm},R)$ is $\gk(\cI)$.
\end{lemma}




\subsection{Approaches from Integer and Linear Programming}

In this section we generalize all the bounds given in \cite{shanmugam2014graph} (which themselves are generalizations of \cite{tehrani2012bipartite}) to the case of the ICCSI problem.
We start with the following definition, introduced in \cite{GCG-cd-art-shum2012broadcasting} as a {\em coding group}, wherein a procedure to detect such as subset is given.
It is easy to see that this definition generalizes the definition of a hyperclique for the ICSI case given in \cite{shanmugam2014graph}.
\begin{definition}
Let $\cI=(\cX,\cX^{(S)},R)$ be an instance of the ICCSI problem. A subset of receivers $C\subseteq [m]$ is called \emph{generalized clique} if there exists ${\bf v}\in\cX^{(S)}$ such that $R_i\in \langle \bv\rangle + \cX^{(i)}$ for all $i\in C$.
\end{definition}

We have the following characterisation of a generalized clique is immediate from the definition.
 
\begin{lemma}
Let $\cI=(\cX,\cX^{(S)},R)$ be an instance of the ICCSI problem. $C\subset [m]$ is a generalized clique if and only if either
of the following equivalent conditions hold:
\begin{enumerate}
	\item there exists ${\bf v}\in\cX^{(S)}$ such that $\langle \bv\rangle \subset \langle R_i \rangle + \cX^{(i)}$ for all $i \in C$, 
	\item $\rk(R_C + A_C) = 1$ for some $m \times n$ matrix $A \in \cX \cap (\oplus \cX^{(S)})$.
\end{enumerate} 
\end{lemma}

For simplicity in the following we refer to a generalized clique just as a clique.

The demand $R_iX$ of each user $i$ of a clique can be met by sending the message $\bv X$ and hence a set of $\ell$ cliques
that partitions the set $[m]$ ensures that all requests can be delivered in at most $\ell$ transmissions. Minimizing this number for a specific instance
can be found via integer programming (see \cite{blasiak2010index,tehrani2012bipartite,shanmugam2014graph}).
Recall that the optimal solution of the LP-relaxation of an IP problem returns rational values.

\begin{definition}
	We denote by $\cC$ the set of all cliques of $\cI=(\cX,\cX^{(S)},R)$. For each clique $C \in \cC$ define the set
	$$
	\cR(C):=\{\bv \in \fq^n \mid R_i\in \langle \bv\rangle + \cX^{(i)}\,\, \forall\,i\in C\}.
	$$
\end{definition}

 ~\\
 \begin{definition}
  We define the \emph{generalized clique cover number} of $\cI$, denoted by $\varphi(\cI)$, to be 
 the optimal solution of the following integer programme:
 $$
 \min \sum_{C\in\cC}y_C
 $$
 $$
 \text{s.t. }\sum_{C:j\in C}y_C=1\text{ for all $j\in[m]$}
 $$
\begin{equation}\label{eq:clique}
y_C\in\{0,1\}\text{ for all }C\in\cC.
\end{equation}

The LP relaxation of \eqref{eq:clique} (so with the relaxed constraint $0 \leq y_C \leq 1$ for all $C$) is the {\em fractional generalized clique cover number} $\varphi_f(\cI)$.
 \end{definition}

\begin{definition}
For each $C\in \cC$ fix a vector $\bv_C\in \cR(C)$. We define the following integer programme with respect to the vectors $\bv_C$. 
 $$
 \min \,k
 $$
 
$$
 \text{s.t. }\sum_{C:\bv_C\notin\cX^{(j)}}y_C\le k \text{ for all $j\in[m]$}
 $$
  $$
\sum_{C:j\in C}y_C=1\text{ for all $j\in[m]$}
 $$
\begin{equation}\label{eq:loclclique}
y_C\in\{0,1\}\text{ for all }C\in\cC \text{ and }k \in\mathbb{N}.
\end{equation}

We denote by $\phi_l(\cI,(\bv_C\in\cR(C):C\in \cC))$ the optimal solution of \eqref{eq:loclclique}, depending on the fixed $\bv_C$'s. The minimum over all possible $\bv_C$'s is called the \emph{local generalized clique cover number}
$$
\varphi_l(\cI)=\mathop{\mbox{min}}_{(\bv_C\in\cR(C):C\in \cC)}\phi_l(\cI,(\bv_C:C\in \cC)).
$$ 
\end{definition}

This is an extension of the local hyperclique cover: 
for a set of fixed $\bv_C$, given user $j \in [m]$ and some feasible solution to \eqref{eq:clique}, count number of cliques $C$ in that generalized clique cover such that $\bv_C$
 is not contained in the side-information $\cX^{(j)}$ and let $k$ be the maximum number of such cliques for each $j$. The optimal solution of (\ref{eq:loclclique}) is the minimum value of $k$ over all possible solutions of  \eqref{eq:clique} and all choices of $\bv_C$.
The minimum of the LP relaxation of \eqref{eq:loclclique} over all possible $\bv_C$'s is called the fractional local generalized clique cover number $\varphi_{lf}(\cI)$.
Both $\varphi_{lf}(\cI)$ and $\varphi_{l}(\cI) $ will be shown to give upper bound on the transmission rate of the instance $\cI$.

\begin{remark}
Consider the \INSr{instance $\cI$ of the ICCSI problem} with $m=n=6$, $\fq=\FF_4=\{0,1,\ga,\ga^2\}$ and $\cX^{(S)}=\FF_4^6$, where $\ga$ is such that $\ga^2=\ga+1$. 
$$
V^{(1)}=\left[\begin{array}{cccccc}
0& 1& 0&0&0&0\\
0&0&0&0& 1& 1\end{array}\right],$$$$
V^{(2)}=\left[\begin{array}{cccccc}
1& 0& 0&0&0&0\\
0&0&1&1& 0& 0\end{array}\right],
$$
$$
V^{(3)}=\left[\begin{array}{cccccc}
0& 0& 0&1&0&0\\
0&0&0&0& 1& 1\end{array}\right],$$$$
V^{(4)}=\left[\begin{array}{cccccc}
0& 0& 1&0&0&0\\
1&1&0&0& 0& 0\end{array}\right],
$$
$$
V^{(5)}=\left[\begin{array}{cccccc}
0& 0& 0&0&0&1\\
0&0&1&1& 0& 0\end{array}\right],$$$$
V^{(6)}=\left[\begin{array}{cccccc}
1& 1& 0&0&0&0\\
0&0&0&0& 1& 0\end{array}\right],
$$
and $R_1=100000$, $R_2=010000$, $R_3=001000$, $R_4=000100$, $R_5=000010$, $R_6=000001$.

Now if we consider the partition $C_1=\{1,2\}$, $C_2=\{3,4\}$, $C_3=\{5,6\}$, and we use 
 $\bv_{C_1}=110000$,
 $\bv_{C_2}=001100$,
 $\bv_{C_3}=00001\ga$,
to encode $X$, then we obtain $k=3$. But using  $\bv_{C_3}=000011$ we have that $k=2$. Clearly the optimal solution of  \eqref{eq:loclclique} depends on the choice of vectors $\bv_C$.
%
\end{remark}

Another approach is based on {\em partition multicast}, as described in \cite{shanmugam2014graph}.

\begin{definition}
We define the \emph{partition generalized multicast number}, $\varphi^p(\cI)$ to be the optimal solution of the following integer program
 $$
 \min \sum_{M \subset [m]}a_Md_M
 $$
 
$$
 \text{s.t. }\sum_{M:j\in M}a_M = 1 \text{ for all $j\in[m]$}
 $$

$$
a_M\in\{0,1\}\text{ for all }M \subset [m], M \neq \emptyset.
$$
\begin{eqnarray}\label{eq:partmulti}
 \text{ and }d_M & = & \dim(\langle R_M\rangle)-\mathop{\mbox{min}}_{j\in M}\dim(\langle R_M\rangle\cap\cX^{(j)}). 
\end{eqnarray}

The LP relaxation of \eqref{eq:partmulti} is called the {\em fractional partition generalized multicast number}, $\varphi_f^p(\cI)$.

\end{definition}

 We remark that $d_M= \mathop{\mbox{max}}_{j\in M} \dim(\langle R_M\rangle /\langle R_M\rangle\cap\cX^{(j)})$.
 We briefly justify the above: each user is assigned to exactly one multicast group $M$, so the selected groups $M$ form a partition of $[m]$. 
 Each member $j$ of a multicast group $M \subset [m]$
 already has access to at least $\dim (\langle R_M \rangle \cap\cX^{(j)})$ independent vectors in $\langle R_M \rangle$. 
 As we'll show in Theorem \ref{th:multi}, a coding scheme can be applied to ensure delivery of all remaining requests within a group using at most $d_M$ transmissions.
 The total number of transmissions required by this scheme is the sum of the $d_M$, over all selected multicast groups $M$. 
 
 The final approach considered combines partition multicast and local clique covering \cite[Definition 10]{shanmugam2014graph}. The users $[m]$ are partitioned into multicast groups and independently covered by generalized cliques. Each multicast group offers a reduce\INSr{d} ICCSI problem, to which a restricted local clique cover is applied. 
 
\begin{definition}

Define the following integer programme 
 $$
 \min \sum_{M\subset[m]} a_M t_M
 $$
 
$$
 \text{s.t. }\sum_{\substack{C:\bv_C\notin\cX^{(j)}\\
 C\cap M\ne \emptyset}}y_C\le t_M \text{ for all $j\in M$}
 $$
  $$
\sum_{M:j\in M}a_M=1,\quad \sum_{C:j\in C}y_C=1\text{ for all $j\in[m]$}
 $$
\begin{equation}\label{eq:loclcliquemulti}
a_M,y_C\in\{0,1\}\text{ for all }C\in\cC,\,M\subset [m] \text{ and }t_M\in\mathbb{N}.
\end{equation}

We denote by $\phi^p_l(\cI,(v_C\in\cR(C):C\in \cC))$ the optimal solution of \eqref{eq:loclcliquemulti} with respect to $(\bv_C\in\cR(C):C\in \cC)$ fixed. The minimum over all possible choices of $\bv_C$ is called the \emph{partitioned local generalized clique cover number}
\[
\varphi^p_l(\cI)=\mathop{\mbox{min}}_{(\bv_C\in\cR(C):C\in \cC)}\phi^p_l(\cI,(\bv_C\in\cR(C):C\in \cC)).
\]

The minimum of the LP relaxation of \eqref{eq:loclcliquemulti} over all possible choices of $\bv_C$ is called the fractional partitioned local generalized clique cover number $\varphi_{lf}^p(\cI)$.

\end{definition}

Now, we will show that achievable schemes exist for all parameters and hence obtain upper bounds on
$\beta(\cI)$. The basic technique is to use MDS codes. 
It will be notationally convenient to express $X$ as a column vector of length $n$ over ${\FF_{q^t}}$. 
We will assume in all cases that $q^t$ is large enough to assure the existence of an $\FF_{q^t}$-MDS code of the required length.

\begin{theorem}\label{thmcliquecover}
	Let $\cI=(\cX,\cX^{(S)},R)$. 
There exist achievable $\fq$-linear index codes corresponding to $\varphi(\cI)$ and $\varphi_f(\cI)$. In particular, we have 
$$\gb(\cI)\le \varphi_f(\cI) \le\varphi(\cI).$$
\end{theorem}
\begin{proof}
For each $C \in \cC$ fix a vector $\bv_C \in \cR(C)$. Then given a clique cover $\cC^{opt} = \{C \in \cC: y_C = 1\}$, corresponding to an optimal solution of \eqref{eq:clique}, 
and a data vector $X$, we broadcast $\{\bv_{C}X : C \in \cC^{opt} \}$. The demands $R_jX$ of each receiver $j \in [m]$  can be met in $|\cC^{opt}| = \varphi(\cI)$ transmissions since $R_j\in   \langle \bv_{C}\rangle + \cX^{(j)}$ for all $j\in C$. 

Now consider the LP relaxation of \eqref{eq:clique} and let an optimal solution be given by $\{y_C :C \in\cC\}\subset  {\mathbb Q}$. 
Let $r$ be the least common denominator of the $y_C$ and for each $C$ define the integral weight $\hat{y}_C = ry_C \in [r]$. 
Denote the resulting multi-set of cliques by $\cC^{opt} = \{ (\hat{y}_C,C) : C \in \cC\}$. 
Every user $j$ is contained in $r$ (not necessarily distinct) cliques of $\cC^{opt}$, with each distinct clique $C$ appearing with multiplicity $\hat{y}_C$.
Now split each packet $X_i \in \FF_{q^t}$ into $r$ packets of equal size, so
consider now $X$ as the data matrix
$$
X=\left[\begin{array}{ccc}
X_1^1  &\dots   & X_1^r\\
\vdots &        & \vdots\\
X_n^1  &\dots   & X_n^r
\end{array}\right],
$$
with coefficients in a subfield ${\mathbb F}_{q^\ell}$ of $\FF_{q^t}$ \INSe{where $\ell$ is the least divisor of $t$ satisfying $r \ell \leq t$.
If $\INSe{q^{\ell}} >s=\sum_C\hat{y}_C$ then there exists an $\FF_{q^\ell}$-$[s,r]$ MDS code, so suppose this is the case and let $G$ be a generator matrix of such a code.} 
Now list the elements of $\cC^{opt}$ as $C_1,...,C_s$ and assign to each column $G^i$ of $G$ the clique $C_i$.

For each clique $C_i$ in $\cC^{opt}$, the packet $\bv_{C_i}XG^i \in \FF_{q^t}$ is transmitted. Each transmission corresponds to an $\fq$-linear combination
of blocks of length $\ell \leq t/r$ over $\fq$ and there are $s \INSe{=r \varphi_f(\cI)}$ transmissions in total. 

Now consider the receiver $j\in[m]$, which has demanded the vector $R_jX$. We may assume that $j$ is contained in the first $r$ cliques $C_1\dots,C_r$ of the list of $s$ cliques. 
Then all users, including $j$, has received $(\bv_{C_1}XG^1,\dots,\bv_{C_r}XG^r) \in \FF_{q^\ell}^r$. 
From Remark \ref{rm:comb} we have $R_j= \ga_i \bv_{C_i}+ \ba_iV^{(j)}$ for some $\ga_i,\ba_i$ for each $i\in[r]$. Thus $j$ can recover the vector 
$$\begin{aligned}
(R_jXG^1,...,R_jXG^r)=&
(\ga_1 \bv_{C_1}XG^1+ \ba_1V^{(j)}XG^1,...,\\
&\ga_1 \bv_{C_1}XG^r+ \ba_1V^{(j)}XG^r)\end{aligned}.
$$
Now
$$
(R_jXG^1,...,R_jXG^r)=R_jX G^{[r]},
$$
where $G^{[r]}=[G^1,...,G^r]$ is an invertible $r \times r$ matrix, by the MDS property of the code generated by $G$. Then $j$ can decode $R_jX$. 
Every user receives the $r$ packets it requires and the total number of transmissions is $s$.

\end{proof}


\begin{theorem}\label{theq4}
Let $\cI=(\cX,\cX^{(S)},R)$. There are achievable linear index codes corresponding to $\varphi_l(\cI)$ and $\varphi_{lf}(\cI)$ implying $\gb(\cI)\le \varphi_{lf}(\cI) \le\varphi_l(\cI)$.
\end{theorem}
\begin{proof}
Let $\cC^{opt}=\{C_1,\dots,C_s\}$ the set of cliques for which
$y_C = 1$ in the optimal solution $(k,\{y_C: C \in \cC\})$ of \eqref{eq:loclclique} for some fixed choice of vectors $\bv_C \in \fq^n$. Let $s=\sum_Cy_C = |\cC^{opt}|$ and
let $G$ the generator matrix of an $\fq$-$[s,k]$ MDS code. 
As before we associate a column of $G$ to each clique in $\cC^{opt}$, and the sender transmits an encoding of the data vector $X \in \FF_{q^t}^{n \times 1}$ as:
$$
Y=\sum_{C \in\cC^{opt}}\bv_{C}XG^{C} =G (\bv_C X)_{C \in\cC^{opt}\bv_C },
$$
which corresponds to $s$ transmissions over $\FF_{q^t}$.
For any $j\in[m]$, the constraints in the integer programme of \eqref{eq:loclclique} require that there are at most $k$ cliques of $\cC^{opt}$ with $\bv_C \notin \cX^j$.
This means that for any choice of $j$, there are at most $k$ vectors in $\{\bv_C: C \in \cC\}$ not contained in $\cX^{(j)}$.
We have
$$
Y=\sum_{C \in\cC^{opt}: \bv_C \in \cX^{(j)} }\bv_{C}XG^{C}+\sum_{C \in\cC^{opt}:\bv_C \notin \cX^j}\bv_{C}XG^{C}.
$$ 
Therefore, Receiver $j$, given its side information $\cX^{(j)}$, can recover 
$$
\sum_{C \in\cC^{opt}:\bv_C \notin \cX^{(j)}}\bv_{C}XG^{C} = \tilde{G}(\bv_C X)_{C \in\cC^{opt}:\bv_C \notin \cX^{(j)}}
$$
where $$\tilde{G}=[G^C]_{C \in\cC^{opt}:\bv_C \notin \cX^j}$$ is a $k \times k$ submatrix of $G$. 
Since $\tilde{G}$ is invertible by the MDS property, the user $j$ can retrieve the vector 
$(\bv_C X)_{C \in\cC^{opt}:\bv_C \notin \cX^{(j)}}$. For a clique $C$ containing $j$, using $\bv_{C}X$ it is possible to retrieve $R_jX$.

Now consider the LP relaxation of \eqref{eq:loclclique} and let $(k,\{y_C: C \in \cC\})$ be an optimal solution for some rationals $0 \leq y_C \leq 1$. This time, let $r$ be the least common denominator of the $y_C$ and $k$ and for each $C$ define $\hat{y}_C = ry_C , \hat{k}=rk\in {\mathbb Z}$. As before, every distinct clique $C$ is assigned an integer weight in $[r]$ and we denote the corresponding multi-set of cliques by $\cC^{opt}$. Every user is contained in $r$ cliques. Let $s=\sum_{C \in \cC}\hat{y}_C$, let $G$ and $H$ be respective generator matrices of $[s,\hat{k}]$ and $[s,r]$ MDS codes over $\FF_{q^t}$. 
Again we represent $X$ as an $n\times r$ matrix with each packet $X_i$ in the form of a vector of length $r$ over a subfield of $\FF_{q^t}$.
Associating the $i$th columns of $G$ and $H$ to the $i$th clique $C_i$ with respect to a fixed listing of the multi-set $\cC^{opt}$, the following is transmitted.
$$
Y=\sum_{i=1}^s(\bv_{C_i}XH^{i})G^{i}.
$$
For any $j\in [m]$, the $j$th receiver uses its side information as before to obtain
$$
\sum_{i=1}^{\hat{k}}(\bv_{C_i}XH^{i})G^{i},
$$
where without loss of generality, $C_1,\dots,C_{\hat{k}}$ are the cliques for which $\bv_C\notin \cX^{(j)}$. Moreover, $j$ is in $r$ of these cliques, which we may suppose to be $C_1,\dots,C_r$.
So as before from the MDS property of $G$, $j$ can recover the vector $(\bv_{C_1}XH^{1},\dots,\bv_{C_{\hat{k}}}XH^{{\hat{k}}})$, and in particular $(\bv_{C_1}XH^{1},\dots,\bv_{C_{{r}}}XH^{r})$.

Since for each $i\in[r]$, $R_j=\ga_i\bv_{C_i}+\ba_iV^{(j)}$ for some $\ga_i$ and $\ba_i$, the user $j$ can obtain
$$
(R_{j}XH^{1},\dots,R_jXH^{r}),
$$
and therefore obtain $R_jX$ by the MDS property of $H$.
Every user receives its required $r$ packets and the total number of transmissions is $\hat{k}$.
\end{proof}

Given an instance $\cI=(\cX,\cX^{(S)},R)$, let $\tilde{m}$ denote the number of distinct equivalence classes of $[m]$ under the relation $i \sim j$ if $\cX^{(i)}=\cX^{(j)}$. We \INSe{will use the following result of \cite{byrne2015error}}, which generalizes Proposition \ref{prop:bound}. 

\begin{proposition}\label{corzip}
Let $\cI=(\cX,\cX^{(S)},R)$. If $q>\tilde{m}$ then $\gk(\cI) \leq \max \{n-d_i :i\in[m]\}$. For any $q$, $\gk(\cI) \leq \rk(R)$.
\end{proposition}

\begin{proof}
	That $\kappa(\cI) \leq \rk(R)$ is trivial: $\kappa(\cI)$ is by definition the miniumum rank of an element of the coset $R + \cX \cap (\oplus \cX^{(S)})$. 
	Indeed, an $\fq$-linear code of length $N=\rk(R)$ exists simply by sending a basis of the rowspace of $R$, in which case no user requires its side-information in order to retrieve its request $R_iX$.
	\INSe{That $\gk(\cI) \leq \max \{n-d_i :i\in[m]\}$ is shown in \cite{byrne2015error}}.
\end{proof}

\INSe{The essential content of the proof of Proposition \ref{corzip} is that there exists an $N \times n$ matrix $L$ realizing $\cI$ for 
	$N\leq \max \{n-d_i :i\in[m]\}$, which corresponds to a multicast solution, so every user can retrieve any linear combination of the $X_i$. In this case the matrix $L$ is such that $\langle L \rangle + \cX^{(i)} = \fq^n$ for each $i$.}

\begin{theorem}\label{th:multi}
Let $\cI=(\cX,\cX^{(S)},R)$. 
There are achievable linear index codes of lengths $\varphi^p(\cI)$ and $\varphi^p_{f}(\INSe{\cI})$, which implies that $\gb(\cI)\le \varphi^p_{f}(\INSe{\cI}) \le\varphi^p(\cI)$.
\end{theorem}
\begin{proof}

Let $\cM$ be a collection of multicast groups $M \subset [m]$ yielding an optimal solution to \eqref{eq:partmulti}. \\
Let $M \in \cM$ and consider the ICCSI instance $\cI_M = (\oplus_{j \in M} \cX^{(j)}, \langle R_M\rangle, R_M)$. 
From Proposition \ref{corzip}, for sufficiently large $q$, there exists 
$L_M \in \fq^{d_M\times n}$ such that each user in $M$ can decode $R_jX$, which uses $d_M$ transmissions. 
Applying this approach to each $M \in \cM$, we find that all users' requests can be retrieved using at most $\varphi^p(\cI)=\sum_{M \in \cM} d_M$ transmissions.

Let us consider now the LP relaxation of \eqref{eq:partmulti} and let $\{a_M : M \subset [m] \} \subset {\mathbb Q}$ be an optimal solution. Let $r$ denote the least common denominator of the $a_M$ and define $\hat{a}_M = ra_M \in {\mathbb Z}$. Every multicast group $M$ is assigned an integer weight in $[r]$ and the multi-set of multicast groups is denoted by $\cM^{opt}$. Every user is contained in $r$ multicast groups of $\cM^{opt}$.
As before, we represent the data vector $X \in \FF_{q^t}^n$ as an $n \times r$ matrix over a subfield of $\FF_{q^t}$.
\INSe{Let 
$L_M$ be an ${d_M\times n}$ matrix 
satisfying $\langle R_M \rangle \subset \langle L_M \rangle + \cX^{(j)}$
for $j \in M$, i.e. such that each user assigned to $M$ can retrieve its requested data $R_jX$.}
Let $s=\sum_M\hat{a}_M$ and, as before, let $G$ be a generator matrix of an $[s,r]$ MDS code over \INSe{$\FF_{q^\ell}$ with $\ell r \leq t$} and associate a column $G^i$ of $G$ to each multicast group $M_i$ in $\cM$.
The sender transmits the $s$ \INSe{$\FF_{q^\ell}$}-vectors of lengths $d_{M_i}$:
$$
L_{M_1}X G^{1},\dots, L_{M_s}X G^{s}.
$$
Let $j \in M_i$ for some $i\in [r]$. User $j$ considers only $r$ vectors, say these are:
$$
L_{M_1}X G^{1},\dots, L_{M_r}X G^{r},
$$
and by assumption can solve for \INSe{some vectors} $\ba_i,\bc_i$ 
$$
R_j=\bc_i L_{M_i}+\ba_i V^{(j)}.
$$
Thus $j$ can recover 
\INSe{
$$R_jXG^i=\bc_i L_{M_i} XG^{i}+\ba_i V^{(j)}XG^{i}$$
as User $j$ knows  $L_{M_i} XG^{i}$ , $V^{(j)}X$ and $G^{i}$.
So, we can compute}
$$
R_jX[G^{1},\dots, G^{r}]
$$
and from the MDS property it is possible to obtain $R_jX$. 
\end{proof}

\begin{remark}
\INSr{Theorem \ref{th:multi} generalizes the statement of \cite[Theorem 2]{shanmugam2014graph}. However, the scheme given in the proof of \cite[Theorem 2]{shanmugam2014graph} to establish the upper bound, is incorrect. We assert that the statement of the theorem is still valid since it is special case of  Theorem \ref{th:multi} and the parameters $\varphi^p$ and $\varphi^p_l$ generalize those given in \cite{shanmugam2014graph}. We provide an example below to show that the scheme proposed in the proof of \cite[Theorem 2]{shanmugam2014graph} does not work.}

\INSr{Consider the instance of the ICSI problem with $m=n=4$, $f(i)=i$ for all $i$ and side information 
$\cX_1=\{2\}$, $\cX_2=\{3,4\}$, $\cX_3=\{1,4\}$ and $\cX_4=\{1,3\}$. The graph $\cG$ associated with this instance is given in Figure \ref{fig:graph51}.}
\begin{figure}[h]\caption{$\cG$}\label{fig:graph51}

\centering
\begin{tikzpicture}[->,>=stealth',shorten >=0.5pt,auto,node distance=1.5cm,
  thick,main node/.style={circle,fill=black!20,draw,font=\sffamily\bfseries}]

  \node[main node] (1) {1};
  \node[main node] (2) [below left of=1] {4};
  \node[main node] (3) [below right of=1] {3};
    \node[main node] (4) [below right of =2] {2};
  \path[every node/.style={font=\sffamily\small}]
   
(1) edge node [left] {}(4)
    (2) edge [right] node [right] {}(1)
        edge [right] node[right] {}(3)
    (3) edge [right] node [right] {}(1)
        edge [right] node[right] {}(2)
    (4) edge [right] node [right] {}(2)
     edge [right] node [right] {}(3)
;
    
\end{tikzpicture} 

 \end{figure}
It can be checked that $\varphi^p(\cG)=3$ and from the LP relaxation we obtain $\varphi^p_f(\cG)=5/2$. Consider for example the set $\cM^{opt}=\{M_1=\{1,2,3\},M_2=\{1,2,4\},M_3=\{3,4\}\}$ arising from an optimal solution of the LP problem. Then $r=2$ and our data matrix is
$$
X=\left[\begin{array}{ccc}
X_1^1&X_1^2\\
X_2^1&X_2^2\\
X_3^1 &X_3^2\\
\INSr{X_4^1} &\INSr{X_4^2}
\end{array}\right].
$$
In \cite{shanmugam2014graph} the authors give the following scheme for the fractional parameter.
We have that every user is contained in $r$ multicast groups (not necessarily differents). Every packet $X_i$ consists of $r$ sub-packets, then we transmit each sub-packet using the scheme corresponding to one of the $r$ multicast groups.

\INSr{For all $i$, denote by $L_i$ the matrix associated to the scheme used to encode the message for the users contained in the set $M_i$. In particular, we can consider the following matrices
$$L_1=\left[\begin{array}{cccc}
1&1&0&0\\
0&1&1&0
\end{array}\right],\,
L_2=\left[\begin{array}{cccc}
1&1&0&0\\
0&1&0&1
\end{array}\right],$$
$$
L_3=\left[\begin{array}{cccc}
0&0&1&1
\end{array}\right].$$ 
Note that only the receivers contained in $M_i$ are able to decode when $L_i$ is used to encode.}

\INSr{Following the scheme given in \cite{shanmugam2014graph}, we do not need to combine the sub-packets $X^1$ and $X^2$, using an MDS code, as in the proof of Theorem \ref{th:multi}.
Therefore, the message transmitted using this scheme will be of type
$$
Y=(L_1X^{i_1},L_2X^{i_2},L_3X^{i_3})
$$
where $i_j\in\{1,2\}$ for each $j$. Thus it should be possible to find a choice of the $i_j$'s such that all the receivers are able to retrieve the requested packet.}

\INSr{Suppose we choose $i_1=1,i_2=2$ and $i_3=1$. Note that in this case the receivers 1, 2 and 4 can retrieve their requested packets, but receiver $3$ obtains only the first sub-packet. It can be checked that for all possible choice of $i_j$, there is at least one receiver that obtains only one of its two requested sub-packets. On the other hand, using an $\FF_2$-$[3,2,2]$ MDS code to combine the sub-packets, we can satisfy all the requests by sending:
$$
Y=(L_1X^{1},L_2X^{2},L_3(X^{1}+X^{2})).
$$}

\end{remark}

\begin{theorem}
There are achievable linear index codes corresponding to $\varphi^p_l(\cI)$ and $\varphi^p_{lf}(\INSe{\cI})$ implying $\gb(\cI)\le \varphi^p_{lf}(\INSe{\cI}) \le\varphi^p_l(\cI)$.
\end{theorem}
\begin{proof}
Fix a set of coding vectors $\{\bv_C \in \cR(C)\}$ for each $C \in \cC$.
Let $\cC^{opt}=\{C_1,\dots,C_s\}$ be the set of cliques for which
$y_C = 1$ in the optimal solution $(\{t_M : M \subset [m]\},\{y_C : C \in \cC\})$ of \eqref{eq:loclcliquemulti}.
Fix a multicast group $M$ and let $G$ be a generator matrix of an $[s,t_M]$ MDS code. Associate each $i$th column of $G$ to the clique $C_i$ in $\cC^{opt}$. 
For this multicast group, the sender transmits  
$$
Y=\sum_{C_i\cap M\ne \emptyset}\bv_{C_i}XG^{i}.
$$
Given the side-information of User $j\in M$ this sum reduces to one involving only $t_M$ cliques, which we may assume to be $C_1,...,C_{t_M}$, yielding
$$
\sum_{i=1}^{t_M}\bv_{C_i}XG^{i}=(\bv_{C_1}X,\dots,\bv_{C_{t_M}}X)[G^{1},\dots,G^{{t_M}}],
$$
and inverting the matrix $[G^{1},\dots,G^{t_M}]$ we can recover $(\bv_{C_1}X,\dots,\bv_{C_{t_M}}X)$. 
As $j$ is contained in one of these cliques it can decode $R_jX$.

Let us consider, now, the LP relaxation of \eqref{eq:loclcliquemulti}. Let $(\{t_M, a_M : M \subset [m]\},\{y_C: C \in \C\})$ be an optimal solution. 
Let $r_1$ denote the least common denominator of the $y_C$ and the $t_M$ and let $r_2$ denote the least common denominator of the $a_M$. 
Define  $\hat{y}_C = r_1y_C$,  $\hat{t}_M = r_1t_M$ and  $\hat{a}_M = r_2a_M$. Every clique $C$ is assigned an integral weight in $[r_1]$ and every multicast group 
$M$ is assigned an integral weight in $[r_2]$. Denote as before the multi-set of cliques by $\cC^{opt}$ and the multi-set of multicast groups by $\cM^{opt}$. 
Every user is contained in $r_1$ cliques and in $r_2$ multicast groups. Moreover, every multicast group in which a user $j$ lies intersects all the $r_1$ cliques related to $j$.
We represent $X$ as an $n \times r_1 r_2$ matrix over a subfield of $\FF_{q^t}$.
Let $s_1=\sum_C\hat{y}_C, s_2=\sum_M\hat{a}_M$ and let $H$ be a generator matrix of an $[s_1s_2,r_1r_2]$ MDS code. We index each column of $H$ by the pair $(k,i)$
associated to a multicast group $M_k$ and clique $C_i$.

Now fix a multicast group $M_k$ and consider a matrix $G$ related to an $[s_1,t_{M_k}]$ MDS code.
The following vector is transmitted:
$$
Y=\sum_{C_i\cap M_k\ne \emptyset}(\bv_{C_i}XH^{(k,i)})G^i.
$$

Let $j\in M_k$. As before, we may assume that, using its side-information, $j$ recovers
$$
\sum_{i=1}^{t_{M_k}}(\bv_{C_i}XH^{(k,i)})G^{i}
$$  
 From the MDS property of the code generated by $G$, $j$ obtains 
$$((\bv_{C_1}XH^{(k,1)}),\dots,(\bv_{C_{t_{M_k}}}XH^{(k,t_M)}).$$ 
 		Restricting to the cliques that contain $j$ we obtain
$$
(\bv_{C_1}XH^{(k,1)},\dots,\bv_{C_{r_1}}XH^{(k,r_1)}).
$$

As $j$ is in $r_2$ multicast groups, without loss of generality $j$ recovers
$$
\begin{aligned}
&(\bv_{C_1}XH^{(1,1)},...,\bv_{C_{r_1}}XH^{(1,r_1)},...,\\
&\bv_{C_1}XH^{(r_2,1)},...,(\bv_{C_{r_1}}XH^{(r_2},r_1)).\end{aligned}
$$
Now using the side information $j$ can compute $R_jX\tilde{H}$
where $$\tilde{H}=[H^{(1,1)},...,H^{(1,r_1)},...,H^{(r_2,1)},...,H^{(r_2,r_1)}].$$ 
From the MDS property of $H$, the receiver $j$ obtains $R_jX$ and hence, $\varphi^p_{lf}$ is achievable.
\end{proof}

\begin{figure}[h!]
\hskip -2.8cm
 \psscalebox{0.65 0.65} 
{
\begin{pspicture}(0,-5.432)(18.048,5.432)
\rput[bl](15.6,-0.968){$\varphi(\mathcal{I})$}
\rput[bl](13.2,3.032){$\psi_f(\mathcal{I})$}
\rput[bl](8.0,1.832){$\psi^p_f(\mathcal{I})$}
\rput[bl](11.6,4.232){$\psi_l(\mathcal{I})$}
\rput[bl](8.8,4.232){$\psi^p_l(\mathcal{I})$}
\rput[bl](5.2,3.032){$\psi^p_{lf}(\mathcal{I})$}
\rput[bl](13.2,1.832){$\psi^p(\mathcal{I})$}
\rput[bl](9.2,3.032){$\psi_{lf}(\mathcal{I})$}
\rput[bl](15.6,3.032){$\psi(\mathcal{I})$}
\rput[bl](13.2,-0.968){$\varphi_f(\mathcal{I})$}
\rput[bl](8.0,-4.968){$\varphi^p_f(\mathcal{I})$}
\rput[bl](11.6,0.232){$\varphi_l(\mathcal{I})$}
\rput[bl](8.8,0.232){$\varphi^p_l(\mathcal{I})$}
\rput[bl](5.2,-0.968){$\varphi^p_{lf}(\mathcal{I})$}
\rput[bl](13.2,-4.968){$\varphi^p(\mathcal{I})$}
\rput[bl](9.2,-0.968){$\varphi_{lf}(\mathcal{I})$}
\pscustom[linecolor=black, linewidth=0.04]
{
\newpath
\moveto(5.2,1.432)
}
\pscustom[linecolor=black, linewidth=0.04]
{
\newpath
\moveto(18.0,1.032)
\lineto(4.4,1.032)
}

\psline[linecolor=black, linewidth=0.04, arrowsize=0.05291667cm 2.0,arrowlength=1.4,arrowinset=0.0]{<-}(6.304,-0.856)(8.500667,-0.856)(8.94,-0.856)(8.94,-0.856)
\psline[linecolor=black, linewidth=0.04, arrowsize=0.05291667cm 2.0,arrowlength=1.4,arrowinset=0.0]{<-}(10.24,-0.888)(13.06,-0.888)
\psline[linecolor=black, linewidth=0.04, arrowsize=0.05291667cm 2.0,arrowlength=1.4,arrowinset=0.0]{<-}(14.136,-0.872)(15.456,-0.872)
\psline[linecolor=black, linewidth=0.04, arrowsize=0.05291667cm 2.0,arrowlength=1.4,arrowinset=0.0]{<-}(9.728,0.35999998)(11.568,0.35999998)(11.568,0.35999998)
\psline[linecolor=black, linewidth=0.04, arrowsize=0.05291667cm 2.0,arrowlength=1.4,arrowinset=0.0]{<-}(12.4,0.232)(15.552,-0.744)
\psline[linecolor=black, linewidth=0.04, arrowsize=0.05291667cm 2.0,arrowlength=1.4,arrowinset=0.0]{<-}(9.0,-4.828)(13.02,-4.868)
\psline[linecolor=black, linewidth=0.04, arrowsize=0.05291667cm 2.0,arrowlength=1.4,arrowinset=0.0]{<-}(6.224,-0.568)(8.8,0.232)
\psline[linecolor=black, linewidth=0.04, arrowsize=0.05291667cm 2.0,arrowlength=1.4,arrowinset=0.0]{<-}(10.22,-0.848)(11.6,0.232)
\rput[bl](10.0,-3.368){$_w\varphi_{lf}(\mathcal{I})$}
\rput[bl](13.2,-3.368){$_w\varphi_{f}(\mathcal{I})$}
\rput[bl](16.0,-3.368){$_w\varphi(\mathcal{I})$}
\rput[bl](10.8,-2.168){$_w\varphi^p_l(\mathcal{I})$}
\rput[bl](13.2,-2.168){$_w\varphi_l(\mathcal{I})$}
\rput[bl](6.8,-3.368){$_w\varphi^p_{lf}(\mathcal{I})$}
\psline[linecolor=black, linewidth=0.04, arrowsize=0.05291667cm 2.0,arrowlength=1.4,arrowinset=0.0]{<-}(14.04,-4.788)(16.4,-3.368)(16.4,-3.368)
\psline[linecolor=black, linewidth=0.04, arrowsize=0.05291667cm 2.0,arrowlength=1.4,arrowinset=0.0]{<-}(16.0,-0.968)(16.4,-2.968)
\psline[linecolor=black, linewidth=0.04, arrowsize=0.05291667cm 2.0,arrowlength=1.4,arrowinset=0.0]{<-}(14.32,-3.272)(15.92,-3.272)
\psline[linecolor=black, linewidth=0.04, arrowsize=0.05291667cm 2.0,arrowlength=1.4,arrowinset=0.0]{<-}(11.264,-3.256)(13.184,-3.256)
\psline[linecolor=black, linewidth=0.04, arrowsize=0.05291667cm 2.0,arrowlength=1.4,arrowinset=0.0]{<-}(7.984,-3.256)(9.904,-3.256)
\psline[linecolor=black, linewidth=0.04, arrowsize=0.05291667cm 2.0,arrowlength=1.4,arrowinset=0.0]{<-}(8.0,-2.968)(10.8,-2.168)
\psline[linecolor=black, linewidth=0.04, arrowsize=0.05291667cm 2.0,arrowlength=1.4,arrowinset=0.0]{<-}(11.872,-2.04)(13.072,-2.04)
\psline[linecolor=black, linewidth=0.04, arrowsize=0.05291667cm 2.0,arrowlength=1.4,arrowinset=0.0]{<-}(14.4,-2.168)(16.048,-3.112)
\psline[linecolor=black, linewidth=0.04, arrowsize=0.05291667cm 2.0,arrowlength=1.4,arrowinset=0.0]{<-}(9.72,0.132)(11.32,-1.708)
\psline[linecolor=black, linewidth=0.04, arrowsize=0.05291667cm 2.0,arrowlength=1.4,arrowinset=0.0]{<-}(12.0,0.232)(13.6,-1.768)
\psline[linecolor=black, linewidth=0.04, arrowsize=0.05291667cm 2.0,arrowlength=1.4,arrowinset=0.0]{<-}(6.0,-0.968)(7.56,-2.868)
\psline[linecolor=black, linewidth=0.04, arrowsize=0.05291667cm 2.0,arrowlength=1.4,arrowinset=0.0]{<-}(6.304,3.192)(9.104,3.192)
\psline[linecolor=black, linewidth=0.04, arrowsize=0.05291667cm 2.0,arrowlength=1.4,arrowinset=0.0]{<-}(6.4,3.432)(8.8,4.232)
\psline[linecolor=black, linewidth=0.04, arrowsize=0.05291667cm 2.0,arrowlength=1.4,arrowinset=0.0]{<-}(9.6,4.352)(11.6,4.352)
\psline[linecolor=black, linewidth=0.04, arrowsize=0.05291667cm 2.0,arrowlength=1.4,arrowinset=0.0]{<-}(10.192,3.16)(12.992,3.16)
\psline[linecolor=black, linewidth=0.04, arrowsize=0.05291667cm 2.0,arrowlength=1.4,arrowinset=0.0]{<-}(14.032,3.144)(15.504,3.144)
\psline[linecolor=black, linewidth=0.04, arrowsize=0.05291667cm 2.0,arrowlength=1.4,arrowinset=0.0]{<-}(14.1,1.992)(15.6,3.032)
\psline[linecolor=black, linewidth=0.04, arrowsize=0.05291667cm 2.0,arrowlength=1.4,arrowinset=0.0]{<-}(9.056,1.96)(13.152,1.96)
\psline[linecolor=black, linewidth=0.04, arrowsize=0.05291667cm 2.0,arrowlength=1.4,arrowinset=0.0]{<-}(6.0,3.032)(7.888,2.088)
\psline[linecolor=black, linewidth=0.04, arrowsize=0.05291667cm 2.0,arrowlength=1.4,arrowinset=0.0]{<-}(10.0,3.432)(11.6,4.232)
\psline[linecolor=black, linewidth=0.04, arrowsize=0.05291667cm 2.0,arrowlength=1.4,arrowinset=0.0]{<-}(12.4,4.232)(15.504,3.352)
\psline[linecolor=black, linewidth=0.04, arrowsize=0.05291667cm 2.0,arrowlength=1.4,arrowinset=0.0]{<-}(8.848,2.076)(13.2,3.032)
\psline[linecolor=black, linewidth=0.04, arrowsize=0.05291667cm 2.0,arrowlength=1.4,arrowinset=0.0]{<-}(8.84,-4.768)(13.2,-3.368)
\psline[linecolor=black, linewidth=0.04, arrowsize=0.05291667cm 2.0,arrowlength=1.4,arrowinset=0.0]{<-}(9.6,-0.968)(10.4,-2.968)
\psline[linecolor=black, linewidth=0.04, arrowsize=0.05291667cm 2.0,arrowlength=1.4,arrowinset=0.0]{<-}(9.64,4.092)(13.136,2.056)
\rput[bl](4.656,4.696){ICSI}
\rput[bl](4.752,0.312){ICCSI}
\psframe[linecolor=black, linewidth=0.04, dimen=outer](18.048,5.368)(4.4,-5.432)
\psline[linecolor=black, linewidth=0.04, arrowsize=0.05291667cm 2.0,arrowlength=1.4,arrowinset=0.0]{<-}(15.536,-5.176)(16.096,-5.144)
\rput[bl](15.264,-5.176){$u$}
\rput[bl](16.208,-5.192){$v$}
\rput[bl](16.816,-5.224){$u\le v$}
\rput[bl](16.448,-5.208){$\equiv$}
\end{pspicture}
}
\caption{The bottom part of the figure describes
ICCSI bounds introduced in this work while
the top describes the ICSI case. Smaller quantities are placed to the left and the weakest
bound is placed to the rightmost of the figure. Arrows
indicate the relationship they satisfy.}
\end{figure}
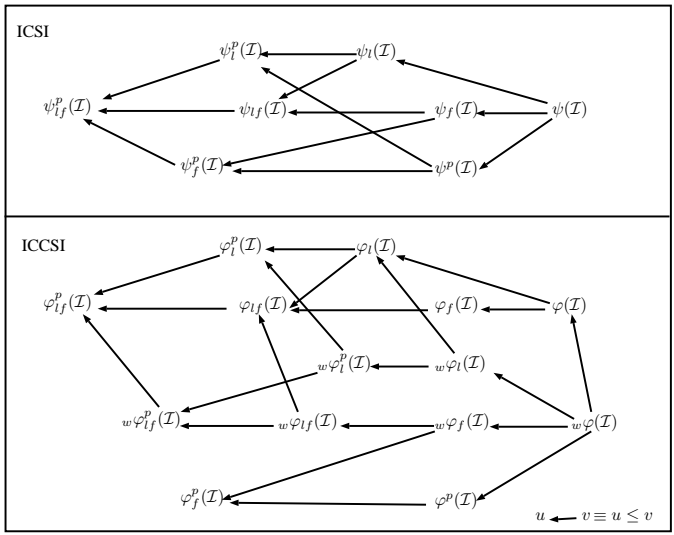
	
\begin{remark}\label{rm:comp1}
The parameters $\varphi^p$ and $\varphi^p_l$ are not comparable. From the parameters given in \cite{shanmugam2014graph} we have that there exist instances of the ICSI problem for which $\varphi^p(\cI)\ge\varphi^p_l(\cI)$. Now consider the ICCSI instance with $m=n=3$, $q=2$, $\cX^{(S)}=\FF_2^3$.
$$
V^{(1)}=[0\; 1\; 1]\quad V^{(2)}=[1\; 1\; 1]\quad V^{(3)}=[1\; 1\; 1],
$$
and $R_1=100$, $R_2=010$, $R_3=001$.

In order to satisfy the requests of a receiver using only one vector then the coding vectors should be 
\begin{itemize}
\item $\bv_1=100$ or $\bv_1'=111$ for User $1$;
\item $\bv_2=010$ or $\bv_2'=101$ for User $2$;
\item $\bv_3=001$ or $\bv_3'=110$ for User $3$.
\end{itemize}
Then the set of all cliques is $\cC=\{\{1\},\{2\},\{3\}\}$. Moreover we can see that $\bv_i,\bv_i'\notin\cX^{(1)}$ for all $i$.
Now if we consider the multicast group $M=\{1,2,3\}$ we can note that $d_M=2$ and that $t_M=3$ because none of the six vectors above is in the space $\cX^{(1)}$. Then we have $2=\varphi^p(\cI)\le\varphi^p_l(\cI)=3$.
\end{remark}

\begin{remark}\label{rm:comp}
The parameters $\varphi^p$ and $\varphi$ are not comparable. From the parameters given in \cite{shanmugam2014graph}, there exist instances of the ICSI problem for which $\varphi(\cI)\ge\varphi^p(\cI)$. Now consider the ICCSI instance with $m=n=2$, $q=2$, $\cX^{(S)}=\FF_2^2$.
$$
V^{(1)}=[1\; 1]\quad V^{(2)}=[0\; 0],
$$
and $R_1=10$, $R_2=01$.
It is easy to check that using the multicast group partition we need two transmissions, but it can be seen that $\{1,2\}$ is a clique and that $\bv_{\{1,2\}}=01 \in \cR(\{1,2\})$, yielding $1=\varphi(\cI)\le\varphi^p(\cI)=2$.
\end{remark}	
	
\begin{remark}
We have $ \varphi^p_{l}(\cI) \le\varphi_{l}(\cI)\le \varphi(\cI)$. It is easy to check that $\varphi_{l}(\cI)\le \varphi(\cI)$ as $t$ is at most equal to the number of cliques that form a partition of $[m]$.
Then we have also $ \varphi^p_{l}(\cI) \le\varphi_{l}(\cI)$. In fact, among the possible optimal solution to obtain we have those where $M=[m]$ and in that case we obtain exactly $ \varphi_{l}(\cI)$.
\end{remark}

\begin{remark}
It is possible to introduce a weak definition of clique. $C\subseteq [m]$ is called weak clique if for all $i,j\in C$ we have $R_j\in\cX^{(i)}$ or $\langle R_j\rangle=\langle R_i\rangle$. Using this definition, it is possible to introduce the notion of a \emph{weak clique cover}, a \emph{local weak clique cover} and a \emph{partitioned local weak clique cover} with respective corresponding parameters 
$_w\varphi(\cI)$, $ _w\varphi_l(\cI)$ and $ _w\varphi^p_l(\cI)$ along with their fractional counterparts.
\end{remark}

\begin{remark}
If $C$ is a weak clique then it is also a generalized clique. We can encode the message using the sum of distinct requests as vector $\bv_C$.
Moreover from the definition of weak clique, if we consider a clique as a multicast group $M$ then it results $d_M=1$. Therefore $\varphi^p(\cI)\le  {_w\varphi(\cI)}$ and the same holds for the fractional parameters. However also in this case the partitioned local weak clique cover and the partitioned multicast cover are not comparable (see example in Remark \ref{rm:comp1}). 
\end{remark}

\bibliographystyle{siamplain}

\end{document}